\documentclass[10pt]{article}
\usepackage{latexsym}
\usepackage{amssymb}
\usepackage{amsthm}
\usepackage{amsmath}
\usepackage{graphicx}
\usepackage[latin1]{inputenc}

\newtheorem{theorem}{Theorem}[section]
\newtheorem{prop}[theorem]{Proposition}
\newtheorem{coro}[theorem]{Corollary}
\newtheorem{lemma}[theorem]{Lemma}
\newtheorem{remark}[theorem]{Remark}

\newcommand{\R}{\mathbb{R}}             
\newcommand{\N}{\mathbb{N}}             
\newcommand{\Z}{\mathbb{Z}}             %
\newcommand{\C}{\mathbb{C}}             
\renewcommand{\H}{\mathcal{H}}          
\newcommand{\B}{\mathcal{B}}            
\newcommand{\M}{\mathcal{M}}            
\renewcommand{\S}{\mathbb{S}}             

\newcommand{\half}{\frac{1}{2}}

\newcommand{\Ga}{\Gamma^1}         
\newcommand{\Gb}{\Gamma^2}         
\newcommand{\Gc}{\Gamma^3}         

\renewcommand{\d}{\partial}            

\newcommand{\fj}{f_j (X, \lambda,z)}

\newcommand{\gj}{g_j (X, \lambda,z)}

\newcommand{\alun}{a_{L1}(\lambda,z)}

\newcommand{\alt}{a_{L3}(\lambda,z)}

\newcommand{\DS}{\mathbb{D}_{\S^2}}

\newcommand{\Section}[1]{\section{#1} \setcounter{equation}{0}}

\author{Thierry Daud\'e \footnote{D\'epartement de math\'ematiques, Universit\'e de Cergy-Pontoise, UMR CNRS 8088, 2 Av. Adolphe Chauvin, 95302 Cergy-Pontoise cedex. Email: thierry.daude@u-cergy.fr. Research  supported by the French National Research Project AARG, No. ANR-12-BS01-012-01}, Damien Gobin and Fran\c{c}ois Nicoleau \footnote{D\'epartement de Math\'ematiques, Universit\'e de Nantes, 2, rue de la Houssini\`ere, BP
    92208, 44322 Nantes Cedex 03. Emails: damien.gobin@univ-nantes.fr, francois.nicoleau@univ-nantes.fr. Research  supported by the French National Research Project NOSEVOL, No. ANR- 2011 BS0101901}}
\title{Local inverse scattering at fixed energy in spherically symmetric asymptotically hyperbolic manifolds}
\date{}

\begin{document}

\maketitle


\begin{abstract}

In this paper, we adapt the well-known \emph{local} uniqueness results of Borg-Marchenko type in the inverse problems for one dimensional Schr\"odinger equation to prove \emph{local} uniqueness results in the setting of inverse \emph{metric} problems. More specifically, we consider a class of spherically symmetric manifolds having two asymptotically hyperbolic ends and study the scattering properties of massless Dirac waves evolving on such manifolds. Using the spherical symmetry of the model, the stationary scattering is encoded by a countable family of one-dimensional Dirac equations. This allows us to define the corresponding transmission coefficients $T(\lambda,n)$ and reflection coefficients $L(\lambda,n)$ and $R(\lambda,n)$ of a Dirac wave having a fixed energy $\lambda$ and angular momentum $n$. For instance, the reflection coefficients $L(\lambda,n)$ correspond to the scattering experiment in which a wave is sent from the \emph{left} end in the remote past and measured in the same left end in the future. The main result of this paper is an inverse uniqueness result local in nature. Namely, we prove that for a fixed  $\lambda \not=0$, the knowledge of the reflection coefficients $L(\lambda,n)$ (resp. $R(\lambda,n)$) - up to a precise error term of the form $O(e^{-2nB})$ with $B>0$ - determines the manifold in a neighbourhood of the left (resp. right) end, the size of this neighbourhood depending on the magnitude $B$ of the error term. The crucial ingredients in the proof of this result are the Complex Angular Momentum method as well as some useful uniqueness results for Laplace transforms.

\vspace{0.5cm}
\end{abstract}


\Section{Introduction and statement of the results}

The aim of this short paper is to extend the \emph{local} inverse uniqueness results of Borg-Marchenko type for one dimensional Schr\"odinger equation obtained first in \cite{Si}, and improved in \cite{Be, GS, Te}, to the setting of inverse \emph{metric} problems, that is inverse problems on three or four dimensional curved manifolds whose unknown - the object we wish to determine by observing waves at infinity - is the (Riemanniann or Lorentzian) metric itself. We shall consider for the moment a very specific and simple class of $3$D-Riemanniann manifolds that we name Spherically Symmetric Asymptotically Hyperbolic Manifolds, in short SSAHM. Precisely, these are described by the set
$$
  \Sigma = \R_x \times \S^2_{\theta,\varphi},
$$
equipped with the Riemanniann metric
$$
  \sigma = dx^2 + a^{-2}(x) \,d\omega^2,
$$
where $d\omega^2 = \left( d\theta^2 + \sin^2\theta \,d\varphi^2 \right)$ is the euclidean metric on the $2$D-sphere $\S^2$. The assumptions on the function $a(x)$ - that determines completely the metric - are:
\begin{equation} \label{RegulA}
  a \in C^2(\R), \quad a > 0,
\end{equation}
and
\begin{equation} \label{AsympA}
  \exists \, a_\pm > 0, \ \kappa_+ < 0, \ \kappa_- > 0, \quad \quad
  \begin{array}{ccc}
    a(x) & = & a_\pm e^{\kappa_\pm x} + O(e^{3\kappa_\pm x}), \quad x \to \pm \infty, \\
    a'(x) & = & a_\pm \kappa_\pm e^{\kappa_\pm x} + O(e^{3\kappa_\pm x}), \quad x \to \pm \infty.
  \end{array}
\end{equation}

Under these assumptions, $(\Sigma,\sigma)$ is clearly a spherically symmetric Riemanniann manifold with two asymptotically hyperbolic ends $\{x = \pm \infty\}$. Note indeed that the metric $\sigma$ is asymptotically a \emph{small} perturbation of the "hyperbolic like" metrics
$$
  \sigma_\pm = dx^2 + e^{-2\kappa_\pm x}  d\omega_\pm^2, \quad x \to \pm \infty,
$$
where $d\omega_\pm^2 = 1/(a_\pm^2) d\omega^2$ are fixed metrics on $\S^2$. From this, we see easily that the sectional curvature of $\sigma$ tends to the \emph{constant negative values} $- (\kappa_\pm)^2$ on the corresponding ends $\{x = \pm \infty\}$. Hence the name "asymptotically hyperbolic" for this kind of geometry. Note in passing that we allow $\kappa_\pm$ to take different values leading to different sectional curvatures in the two ends. We emphasize at last that such SSAHM are very particular cases (because of our assumption of spherical symmetry) of the much broader class of asymptotically hyperbolic manifolds for instance described in \cite{IK, JSB, SB} (to cite only a few papers that deal with inverse problems).

On the manifold $(\Sigma,\sigma)$, we are interested in studying how (scalar, electromagnetic, Dirac, \dots ) waves evolve, scatter at late times and ultimately, in trying to answer the question: can we determine the metric by observing these waves at the infinities of the manifold (in our model, the two ends $\{x=\pm\infty\}$). For definiteness, we shall consider in this paper how massless Dirac waves propagate and scatter towards the two asymptotically hyperbolic ends. Note that the same results should hold with the Dirac equation replaced by the wave equation. Precisely, let us consider the massless Dirac equation
\begin{equation} \label{DE}
  i \partial_t \psi = \mathbb{D}_{\sigma} \psi,
\end{equation}
where $\mathbb{D}_{\sigma}$ denotes a representation of the Dirac operator on $(\Sigma,\sigma)$ and the $2$-spinor solution $\psi$ belongs to $L^2(\Sigma; \C^2)$. It will be shown in Section \ref{TM} that we have a very simple connection between $\mathbb{D}_{\sigma}$ and the function $a(x)$ appearing in the metric $\sigma$, precisely
\begin{equation} \label{DO}
  \mathbb{D}_{\sigma} = \Ga D_x + a(x) \mathbb{D}_{\S^2},
\end{equation}
where $\mathbb{D}_{\S^2}$ denotes the intrinsic Dirac operator on $\S^2$, represented here by the expression
\begin{equation} \label{DiracS2}
  \mathbb{D}_{\S^2} = \Gb \left( D_{\theta} + \frac{i \cot{\theta}}{2} \right) + \Gc \frac{1}{\sin{\theta}} D_{\varphi},
\end{equation}
with $D_x = -i\partial_x, \ D_\theta = -i\partial_\theta, \ D_\varphi = -i\partial_\varphi$ and where the $2 \times 2$- Dirac matrices $\Ga, \Gb, \Gc$ satisfy the usual anti-commutation relations
\begin{equation} \label{Anticom}
  \Gamma^i \Gamma^j + \Gamma^j \Gamma^i = 2 \delta_{ij}.
\end{equation}

Due to the spherical symmetry of the problem and the existence of generalized spherical harmonics $\{Y_{kl}\}$ that "diagonalize" $\mathbb{D}_{\S^2}$, we can decompose the energy Hilbert space $\H = L^2(\Sigma; \C^2)$ onto a Hilbert sum of partial Hilbert spaces $\H_{kl}$ with the property that the $\H_{kl}$'s are let invariant through the action of the Dirac operator (\ref{DO}). More precisely, if we introduce the set of indices $I = \{ k \in 1/2 + \Z, \ l \in 1/2 + \N, \ |k| \leq l\}$, we have
$$
  \H = \oplus_{kl} \H_{kl}, \quad \H_{kl} = L^2(\R;\C^2) \otimes Y_{kl},
$$
and
$$
  \mathbb{D}_{\sigma}^{kl} := \mathbb{D}_{\sigma \ |\H_{kl}} = \Ga D_x - (l+1/2) a(x) \Gb.
$$
Note that the partial Dirac operator $\mathbb{D}_{\sigma}^{kl}$'s only depend on the \emph{angular momentum} $l + 1/2 \in \N^*$.
For simplicity, we shall denote $l+\half$ by $n$ (hence the new parameter $n$ runs over the integers $\N^*$) and also
\begin{equation} \label{DO-1D}
  \mathbb{D}_{\sigma}^n = \Ga D_x - n a(x) \Gb,
\end{equation}
for the partial Dirac operators on each generalized spherical harmonic $Y_{kl}$.

We are thus led to consider the restriction of the Dirac equation (\ref{DE}) to each partial Hilbert space $\H_{kl}$ separatly and study the properties of the family of $1$D Dirac Hamiltonians $\mathbb{D}_{\sigma}^n, \ n \in \N^*$ in order to obtain spectral, direct and inverse scattering results for the complete Dirac Hamiltonian $\mathbb{D}_{\sigma}$. This has been done in \cite{DN3} in a very similar context\footnote{Note that the models studied in \cite{DN3} (see also \cite{Da1,DN2}) come from General Relativity. More precisely, the direct and inverse scattering of Dirac waves propagating in the exterior region of a Reissner-Nordstr\"om-de Sitter black holes were studied therein. It turns out that these  "relativistic" models and the one presented in this paper are equivalent. This was briefly mentioned in \cite{DN3} and made rigorous in the next Section \ref{TM}.} (see also \cite{AKM, Da1, DN2}). Let us summarize here these results. We refer to Section \ref{SS} for more explanations.

First, the Dirac Hamiltonian $\mathbb{D}_{\sigma}$ is selfadjoint on the Hilbert space $\H = L^2(\Sigma; \C^2)$ and has absolutely continuous spectrum. In particular, the pure point spectrum of $\mathbb{D}_{\sigma}$ is empty. As a consequence, the energy of massless Dirac fields cannot remain trapped on any compact subsets of $\Sigma$, \textit{i.e.} for all compact subset $K \subset \R$,
$$
  \lim_{t \to \pm \infty} \| \mathbf{1}_{K}(x) e^{-it\mathbb{D}_{\sigma}} \psi \| = 0.
$$
In other words, the massless Dirac fields \emph{scatter} towards the asymptotic ends $\{ x = \pm \infty \}$ of the manifold $\Sigma$ at late times.

Second, a complete direct scattering theory can be established for $\mathbb{D}_{\sigma}$ on $(\Sigma,\sigma)$. For all energy $\lambda \in \R$, we denote the scattering matrix at energy $\lambda$ by $S(\lambda)$. It is a unitary operator on $L^2(\S^2; \C^2)$ and thus has the structure of an operator valued $2 \times 2$ matrix, \textit{i.e.}
\begin{equation} \label{SM}
  S(\lambda) = \left[ \begin{array}{cc} T_L(\lambda) & R(\lambda)\\ L(\lambda)&T_R(\lambda) \end{array} \right],
\end{equation}
where $T_L, T_R$ are the transmission operators and $R, L$ the reflection operators. The formers measure the part of a signal having energy $\lambda$ transmitted from an end to the other end in a scattering process whereas the latters measure the part of a signal of energy $\lambda$ reflected from an end to itself ($\{x = -\infty\}$ for $L$ and $\{ x = +\infty\}$ for $R$).

Due to the spherical symmetry of the model, the scattering matrix lets invariant all the partial Hilbert spaces $\H_{kl}$ and can be thus decomposed into a Hilbert sum of unitary operators acting $\C^2$. We write as a shorthand
$$
  S(\lambda) = \sum_{k,l \in I} S_{kl}(\lambda), \quad S_{kl}(\lambda) := S(\lambda)_{|\H_{kl}}.
$$
Since the $1$D Dirac operator (\ref{DO-1D}) only depends on $n = l+1/2 \in \N^*$, the partial scattering matrices $S_{kl}(\lambda)$ also only depend on $n$. We shall thus use the notation
\begin{equation} \label{SM-n}
  S(\lambda,n) = \left[ \begin{array}{cc} T(\lambda,n)&R(\lambda,n)\\ L(\lambda,n)&T(\lambda,n) \end{array} \right].
\end{equation}
For all $n \in \N^*$, we emphasize that the partial scattering matrices $S(\lambda,n)$ are unitary matrices that encode the stationary scattering at a fixed energy $\lambda$ on a given generalized spherical harmonics $\H_{kl}$ with $n = l + 1/2$. As above, the transmission coefficients $T(\lambda,n)$ correspond to the transmitted part of a signal (from one end to the other) whereas the left $L(\lambda,n)$ and right $R(\lambda,n)$ reflection coefficients correspond to the reflected part of a signal in a given end.

In \cite{DN3}, we addressed the question whether it was possible to determine uniquely the metric from the knowledge
of the reflection coefficients $L(\lambda,n)$ or $R(\lambda,n)$. Using essentially the Complex Angular Momentum method (see \cite{Re} for the first appearance of this method and \cite{Ra} for an application to Schr\"odinger inverse scattering), we were able to answer positively to the question with some interesting improvements in the hypotheses. Precisely, we state here the inverse scattering uniqueness result proved in \cite{DN3}.

\begin{theorem}\label{globaluniqueness}
Let $\Sigma = \R \times \S^2$ be a SSAHM equipped with the Riemanniann metric
$$
  \sigma = dx^2 + a^{-2}(x) d\omega^2,
$$
where the function $a(x)$ satisfies the assumptions (\ref{RegulA}) - (\ref{AsympA}). Let $\mathbb{D}_{\sigma} = \Ga D_x + a(x) \mathbb{D}_{\S^2}$ be an expression of the massless Dirac operator associated to $(\Sigma,\sigma)$. To the evolution equation $i\partial_t \psi = \mathbb{D}_{\sigma} \psi$ with $\psi \in \H = L^2(\Sigma; \C^2)$, we associate the countable family of partial waves scattering matrices $S(\lambda,n)$ for $\lambda \in \R$ and $n \in \N^*$ as above. Consider also a subset $\mathcal{L}$ of $\N^*$ that satisfies a M\"untz condition
$$
  \sum_{n \in \mathcal{L}} \frac{1}{n} = \infty.
$$
Then the knowledge of either $R(\lambda,n)$ or $L(\lambda,n)$ for a fixed $\lambda \ne 0$ and for all $n \in \mathcal{L}$ determines uniquely the function $a(x)$ (and thus the metric $\sigma$) up to a discrete set of translations.
\end{theorem}

\begin{remark}
First, we emphasize that the above result is not true if $\lambda=0$, (see Remark 3.7, \cite{DN3}). Secondly, in \cite{DN3}, Theorem 1.1, it is claimed that the knowledge of the transmission coefficients $T(\lambda,n)$ for a fixed $\lambda \ne 0$ and for all $n \in \mathcal{L}$ also
determines uniquely the function $a(x)$ up to a translation. The crucial ingredient of the proof can be found in the Proposition 3.13 of \cite{DN3} which states
that if $T(\lambda,n) = \tilde{T}(\lambda,n)$ for all $n \in \mathcal{L}$, then the corresponding reflection coefficients $L(\lambda,n)$ and $\tilde{L}(\lambda,n)$
(resp. $R(\lambda,n)$ and $\tilde{R}(\lambda,n)$) coincide up to a multiplicative constant. The proof of this result given in \cite{DN3} is unfortunately incomplete and therefore,
this last point is not so clear and could even be false.  However, in this paper, in Proposition \ref{unicitetransmission}, Addendum B, we prove that the knowledge of the transmission coefficients $T(\lambda,n)$
for all $n \in \mathcal{L}$ together with the knowledge of the reflection coefficients
$L(\lambda,k)$ for a finite number of integer $k$, (and a technical assumption on the sectional curvatures), uniquely determines the function $a(x)$ up to a translation. The question whether these last hypotheseses are necessary remains open.
\end{remark}

For more general Asymptotically Hyperbolic Manifolds (AHM in short) with no particular symmetry, difficult direct and inverse scattering results for \emph{scalar waves} have been proved by Joshi, S\'a Barreto in \cite{JSB}, by S\'a Barreto in \cite{SB} and by Isozaki, Kurylev in \cite{IK}. In \cite{JSB} for instance, it is shown that the asymptotics of the metric of an AHM are uniquely determined (up to certain diffeomorphisms) by the scattering matrix $S(\lambda)$ at a fixed energy $\lambda$ off a discrete subset of $\R$. In \cite{SB}, it is proved that the metric of an AHM is uniquely determined (up to certain diffeomorphisms) by the scattering matrix $S(\lambda)$ for every $\lambda \in \R$ off an exceptional subset. Similar results are obtained recently in \cite{IK} for even more general classes of AHM. At last, we also mention \cite{BP} where related inverse problems - inverse resonance problems - are studied in certain subclasses of AHM.

The new inverse scattering results of this paper are local in nature, in the same spirit as \cite{Be, GS, Si, Te}. Instead of assuming the full knowledge of one of the reflection operators, we instead assume the knowledge of one of these operators up to some precise error remainder (see below). Using the particular analytic properties of the scattering coefficients $ L(\lambda,z)$ and $R(\lambda,z)$ with respect to the complex angular momentum $z$ and some well known uniqueness properties of the Laplace transform (see \cite{Ho1, Si}), we are able to prove the following improvement of our previous result.

\begin{theorem} \label{Reflection}
Let $(\Sigma, \sigma)$ and $(\Sigma,\tilde{\sigma})$ two a priori different SSAHM. We denote by $a(x)$ and $\tilde{a}(x)$ the two radial functions defining the metrics $\sigma$ and $\tilde{\sigma}$. We define
$$
  A = \int_\R a(x) dx, \quad \tilde{A} = \int_\R \tilde{a}(x) dx,
$$
as well as the diffeomorphisms
$$
  \begin{array}{ccl} g: \ \R & \longrightarrow & (0, A), \\
                          x & \longrightarrow & g(x) = \int_{-\infty}^x a(s) ds, \end{array}, \quad \quad
  \begin{array}{ccl} \tilde{g}: \ \R & \longrightarrow & (0, \tilde{A}), \\
                          x & \longrightarrow & \tilde{g}(x) = \int_{-\infty}^x \tilde{a}(s) ds, \end{array}.
$$
We also denote by $h=g^{-1} : (0,A) \longrightarrow \R,  \ \tilde{h}=\tilde{g}^{-1} : (0,\tilde{A}) \longrightarrow \R$ their inverse diffeomorphisms. As above, we define $S(\lambda,n)$ and $\tilde{S}(\lambda,n)$ the corresponding partial scattering matrices. Let $\lambda\not= 0$ be a fixed energy and $0< B <
\min\ (A, \tilde{A})$. Then the following assertions are equivalent:

\begin{equation} \label{Ln}
  (i) \quad L(\lambda,n) = \tilde{L}(\lambda,n) + \ O\left(e^{-2nB} \right),\ n \rightarrow +\infty.
\end{equation}
$$
  (ii) \quad \exists k \in \Z, \quad a(x) = \tilde{a}(x + \frac{k\pi}{\lambda}), \quad \forall \ x  \leq h(B)= \tilde{h}(B) - \frac{k\pi}{\lambda}.
$$
Symmetrically, the following assertions are also equivalent:
\begin{equation} \label{Rn}
  (iii) \quad R(\lambda,n) = \tilde{R}(\lambda,n) + \ O\left(e^{-2nB} \right), \ n \rightarrow +\infty.
\end{equation}
$$
  (iv) \quad \exists k \in \Z, \quad a(x) = \tilde{a}(x + \frac{k\pi}{\lambda}), \quad \forall \ x  \geq h(A-B)= \tilde{h}(\tilde{A}-B) - \frac{k\pi}{\lambda}.
$$

\end{theorem}

The above result asserts that the \emph{partial} knowledge of the reflection coefficients in the sense of (\ref{Ln}) or (\ref{Rn}) allows to determine uniquely the metric $\sigma$ in the neighbourhoods of the two ends $\{x = \pm \infty\}$. The size of these neighbourhoods depend on the magnitude of the error terms in (\ref{Ln}) - (\ref{Rn}). Of course, $h(B)$ (resp. $h(A-B)$) depends on the metric $a(x)$ which is a priori unknown. But, it is not difficult to prove using (\ref{AsympA}) that
$h(B) \sim \frac{1}{\kappa_-} \ \log B$ when $B \rightarrow 0$. In the same way, $h(A-B) \sim - \frac{1}{\kappa_+} \ \log (A-B)$ when $B \rightarrow  A$. We also emphasize that the "surface gravities" $\kappa_{\pm}$ can be explicitly recover from the asymptotics of $L(\lambda, n)$ or $R(\lambda,n)$, $n \rightarrow + \infty$,
(see \cite{DN3}, Theorem 4.22).

\vspace{0.3cm}

As a direct consequence, we obtain immediately the following \emph{global uniqueness} result for the metric of a SSAHM. This result slightly improves our earlier version obtained in \cite{DN3} and stated in Theorem \ref{globaluniqueness}.

\begin{coro}
Assume that for a given $C \geq \min (A, \tilde{A})$ and $\lambda\not=0$ a fixed energy, one of the following assertions holds :
$$
 (i) \quad L(\lambda,n) = \tilde{L}(\lambda,n) + \ O\left(e^{-2nC} \right),\ n \rightarrow +\infty.
$$
$$
  (ii) \quad R(\lambda,n) = \tilde{R}(\lambda,n) + \ O\left(e^{-2nC} \right) ,\ n \rightarrow +\infty.
$$
Then, there  exists $k \in \Z, \ a(x) = \tilde{a}(x+ \frac{k\pi}{\lambda}), \ \forall \ x \in \R.$
\end{coro}

\begin{proof} Let us treat for instance the case $(i)$ and assume that $A \leq \tilde{A}$. From our hypothesis, for all $B<A$, we have
$$
L(\lambda,n) = \tilde{L}(\lambda,n) + \ O\left(e^{-2nB} \right).
$$
Hence Theorem \ref{Reflection} implies that $a(x) = \tilde{a}(x + \frac{k\pi}{\lambda}), \ \forall \ x  \leq h(B)$. The result follows letting $B$ tend to $A$ and using that ${\displaystyle{\lim_{X \to A} h(X) = +\infty}}$. Note that we also obtain  $A= \tilde{A}$.
\end{proof}

Let us give here a possible interpretation of the above \emph{local uniqueness} result Theorem \ref{Reflection}. Consider for instance the reflection coefficients $L(\lambda,n)$ and recall that it encodes the following scattering experiment: a wave having energy $\lambda$ is sent from the end $\{x = - \infty\}$ in the past and evolves on the SSAHM. Then $L(\lambda,n)$ measures the part of this wave that is reflected to the same end $\{ x = -\infty \}$ in the far future. Now our result asserts that if we know $L(\lambda,n)$ up to a precise error term of the form $O \left(e^{-2nB}\right)$, then the metric is uniquely determined in a neighbourhood of $\{x = -\infty\}$, the size of the neighbourhood depending only on the constant $B$ defining the error term. We infer thus that, under our assumption, the wave sent from $\{ x = -\infty \}$ hasn't the time to travel through the whole manifold before being measured back in the end $\{ x = - \infty\}$. This explains heuristically why the partial knowledge of $L(\lambda,n)$, in the precise sense given by our assumption, is not enough to determine the full metric. \\

At last, when using the transmission coefficients as the starting point of our inverse problem, we get a result different in nature than the one obtained with the reflection coefficients. Precisely, we obtain a \emph{global uniqueness} result. Moreover, as we have said before,  we have to assume that the reflection coefficients $L(\lambda,n)$ are equal for a finite number of integer $n$ and we make a technical assumption on the sectional curvatures $\kappa_\pm$. The question whether these last hypotheses are necessary remains open.

\begin{theorem} \label{Transmission}
Assume that
$$
	\frac{1}{\kappa_+} + \frac{1}{\kappa_-} < 0, \quad \frac{1}{\tilde{\kappa}_+} + \frac{1}{\tilde{\kappa}_-} < 0,
$$
and that for a fixed energy $\lambda\not=0$ and for some $B > \max(A,\tilde{A})$,
\begin{equation} \label{Tn}
   \quad T(\lambda,n) = \tilde{T}(\lambda,n) + \ O\left(e^{-2nB} \right),\ n \rightarrow +\infty.
\end{equation}
Assume also that
\begin{equation} \label{rr11}
  L(\lambda,n) = \tilde{L}(\lambda,n),
\end{equation}	
for a finite but large enough number of indices $n \in \N$. Then there exists a constant $\sigma \in \R$ such that
$$
  \tilde{a}(x) = a(x + \sigma).
$$
In consequence, the two SSAHM $(\Sigma,g)$ and $(\tilde{\Sigma},\tilde{g})$ coincide up to isometries.
\end{theorem}

Let us make a few comments on this result. First, we provide an heuristic reason why we don't have a \emph{local uniqueness} result when we assume the knowledge of the transmission coefficients up to a precise error. The transmission coefficients - by definition - measure the part of a wave transmitted from one end, say $\{x = -\infty\}$, to the other end $\{x = +\infty\}$. In our case where the SSAHM has only two ends, the transmitted wave has thus the time to propagate into the whole manifold. It is then natural that the transmission coefficients encode \emph{all} the information of the SSAHM.

Second, the asymptotics of the transmission coefficients when $n$ tends to infinity are computed in \cite{DN3}. Precisely, we have
$$
  \mid T(\lambda,n) \mid \sim C \ e^{-nA}, \quad \mid \tilde{T}(\lambda,n) \mid \sim \tilde{C} \ e^{-n\tilde{A}}.
$$
Hence, the condition on $B$ in Theorem \ref{Transmission} cannot be weaker that $B > \frac{1}{2} \max (A, \tilde{A})$. Note then that the assertion (i) implies immediately $A = \tilde{A}$ from the above asymptotics. We mention that a global uniqueness inverse result in the case where $\frac{1}{2} \max (A, \tilde{A}) < B \leq \max (A, \tilde{A})$ is still an open question.

\vspace{0.5cm}

This paper is organised as follows. In Section \ref{TM}, we recall how to compute the Dirac equation on a curved manifold and apply this formalism to obtain a representation of a massless Dirac operator on a SSAHM that is suitable for us. In Section \ref{SS}, we recall the main results from \cite{DN3} where a complete description of the stationary scattering corresponding to massless Dirac fields evolving in a SSAHM was obtained. In Section \ref{LIS}, we prove our main results, Theorem \ref{Reflection}, Theorem \ref{Transmission}. Eventually, we include the last Section \ref{BH} in which an application of our local inverse uniqueness results on SSAHM is given in the context of black hole spacetimes, precisely on Reissner-Nordstr\"om-de-Sitter black holes.


\Section{The model} \label{TM}

Since the Dirac equation is by essence a relativistic equation, we prefer to work directly on the four dimensional Lorentzian manifold $(M,\tau)$ defined by
$$
  M = \R_t \times \Sigma,
$$
and equipped with the metric
$$
  \tau = dt^2 - \sigma,
$$
where $(\Sigma,\sigma)$ is the SSAHM we aim to study. Below we recall how to compute the massless Dirac equation on such a $4$D curved background and obtain a representation of it that we put under the generic Hamiltonian form
$$
  i \partial_t \psi = \mathbb{D}_{\tau} \psi.
$$
We shall see in Remark \ref{Link} that this procedure leads to an equivalent form of the massless Dirac equation than the one
$$
  i \partial_t \psi = \mathbb{D}_{\sigma} \psi,
$$
we would have obtained working on the $3$D-Riemanniann manifold $(\Sigma,\sigma)$. In other words, the Dirac operators $\mathbb{D}_{\tau}$ and $\mathbb{D}_{\sigma}$ that we obtain by these two formalisms are shown to be unitarily equivalent. We prefer to work with the relativistic point of view nevertheless since we are also interested in applications of our inverse results to spacetimes coming from General Relativity, namely black hole spacetimes. We postpone this parenthesis till Section \ref{BH}.

\subsection{Orthonormal frame formalism for the massless Dirac equation in $4$D curved space-time}

To calculate the massless Dirac equation in a 4D curved spacetime $M$ equipped with a Lorentzian metric $\tau$ of signature $(1,-1,-1,-1)$, we use Cartan's orthonormal frame formalism as explained for instance in \cite{CP} or in a more relativistic setting in \cite{Ni}. Let us denote by $\{e_A\}_{A=0,1,2,3}$ a given local Lorentz frame, \textit{i.e.} a set of vector fields satisfying $\tau(e_A, e_B) = \eta_{AB}$ where $\eta_{AB} =$ diag$(1,-1,-1,-1)$ is the flat (Lorentz) metric. We also denote by $\{e^A\}_{A=0,1,2,3}$ the set of dual $1$-forms of the frame $\{e_A\}$. Latin letters A,B will denote in what follows local Lorentz frame indices, while Greek letters $\mu, \nu$ run over four-dimensional space-time coordinate indices. The \emph{massless} Dirac equation takes then the generic form
\begin{equation} \label{AbstractDiracEq}
  \mathbb{D} \phi = \gamma^A (\partial_A + \Gamma_A) \phi = 0.
\end{equation}
Here, the $\gamma^A$'s are the gamma Dirac matrices satisfying the anticommutation relations
\begin{equation} \label{Clifford}
  \{\gamma^A, \gamma^B\} = \gamma^A \gamma^B + \gamma^B \gamma^A = 2 \eta^{AB}.
\end{equation}
The differential operators $\partial_A$'s are given by $\partial_A = e_A^\mu \partial_\mu$ in terms of the local differential operators and the $\Gamma_A$'s are the components of the spinor connection $\Gamma = \Gamma_A e^A = \Gamma_\mu dx^\mu$ in the local Lorentz frame. In order to derive the latter, we first compute the spin-connection $1$-form $\omega_{AB} = \omega_{AB\mu} dx^\mu = f_{ABC} e^C$ thanks to Cartan's first structural equation and the skew-symmetric condition
\begin{equation} \label{Cartan}
  d e^A + \omega^A_{\ B} \wedge e^B = 0, \quad \omega_{AB} = \eta_{AC} \omega^C_{\ B} = - \omega_{BA}.
\end{equation}
Note here that we use the flat metric $\eta_{AB}$ or its inverse $\eta^{AB}$ to raise or lower Latin indices. We also use Einstein summation convention. With this definition, the spinor connection $\Gamma$ is then defined as
\begin{equation} \label{Gamma}
  \Gamma = \frac{1}{8} [\gamma^A, \gamma^B] \omega_{AB} = \frac{1}{4} \gamma^A \gamma^B \omega_{AB} = \frac{1}{4} \gamma^A \gamma^B f_{ABC} e^C.
\end{equation}

\subsection{The Dirac equation on a SSAHM}

We now apply this formalism to calculate the massless Dirac equation on the $4$D- Lorentzian manifold $(M,\tau)$ given by
$$
  M = \R_t \times \R_x \times \S^2_{\theta,\varphi},
$$
and
$$
  \tau = dt^2 - dx^2 - a^{-2}(x) \left( d\theta^2 + \sin^2\theta \,d\varphi^2 \right),
$$
where
$$
  a \in C^2(\R), \quad a > 0.
$$

The spherical symmetry of the metric leads to the natural choice of local Lorentz frame
\begin{equation} \label{LocalFrame}
  e_0 = \partial_t, \quad e_1 = \partial_x, \quad e_2 = a(x) \partial_\theta, \quad e_3 = \frac{a(x)}{\sin\theta} \partial_\varphi.
\end{equation}
The dual $1$-forms are then given by
\begin{equation} \label{LocalDualFrame}
  e^0 = dt, \quad e^1 = dx, \quad e^2 = a^{-1}(x) d\theta, \quad e^3 = \frac{\sin\theta}{a(x)} d\varphi.
\end{equation}
The exterior derivatives of the $e^A$'s are readily computed
$$
  d e^0 = 0, \quad de^1 = 0, \quad de^2 = \frac{a'(x)}{a(x)} \, e^2 \wedge e^1, \quad de^3 = \frac{a'(x)}{a(x)} \, e^3 \wedge e^1 - a(x) \cot\theta \, e^3 \wedge e^2.
$$
Using (\ref{Cartan}), we then easily get
$$
  \omega^0_{\ 1} = \omega^0_{\ 2} = \omega^0_{\ 3} = 0, \quad \omega^1_{\ 2} = \frac{a'(x)}{a(x)} \, e^2, \quad \omega^1_{\ 3} = \frac{a'(x)}{a(x)} \, e^3, \quad \omega^2_{\ 3} = - a(x) \cot \theta \, e^3,
$$
or equivalently
$$
  \omega_{01} = \omega_{02} = \omega_{03} =0, \quad \omega_{12} = - \frac{a'(x)}{a(x)} \, e^2, \quad \omega_{13} = - \frac{a'(x)}{a(x)} \, e^3, \quad \omega_{23} = a(x) \cot \theta \, e^3.
$$
Hence we deduce
\begin{equation} \label{SpinConnection}
  \Gamma = \frac{1}{4} \gamma^A \gamma^B \omega_{AB} = \left( - \frac{a'(x)}{2 a(x)} \gamma^1 \gamma^2 \right) e^2 + \left( - \frac{a'(x)}{2 a(x)} \gamma^1 \gamma^3 + \frac{a(x) \cot \theta}{2} \, \gamma^2 \gamma^3 \right) e^3= \Gamma_A e^A.
\end{equation}

The massless Dirac equation $\gamma^A (\partial_A + \Gamma_A) \phi = 0$ on $(M,\tau)$ thus takes the form
$$
  \left[ \gamma^0 \partial_t + \gamma^1 \partial_x + \gamma^2 \left( a(x) \partial_\theta - \frac{a'(x)}{2 a(x)} \gamma^1 \gamma^2 \right) + \gamma^3 \left( \frac{a(x)}{\sin \theta} \partial_\varphi - \frac{a'(x)}{2 a(x)} \gamma^1 \gamma^3 + \frac{a(x) \cot \theta}{2} \, \gamma^2 \gamma^3 \right) \right] \phi = 0,
$$
or using (\ref{Clifford})
$$
  \left[ \gamma^0 \partial_t + \gamma^1 \partial_x + a(x) \left( \left( \partial_\theta - \frac{\cot \theta}{2} \right) \gamma^2 + \frac{1}{\sin \theta} \partial_\varphi \gamma^3 \right) - \frac{a'(x)}{a(x)} \gamma^1 \right] \phi = 0.
$$
We can get rid of some potentials by considering the weighted spinor
\begin{equation} \label{WeightedSpinor}
  \psi = a^{-1}(x) \phi.
\end{equation}
Then $\psi$ satifies the equation
$$
  \left[ \gamma^0 \partial_t + \gamma^1 \partial_x + a(x) \left( \left( \partial_\theta - \frac{\cot \theta}{2} \right) \gamma^2 + \frac{1}{\sin \theta} \partial_\varphi \gamma^3 \right) \right] \psi = 0.
$$
We finally put this equation under Hamiltonian form. The spinor $\psi$ thus satisfies
$$
  i \partial_t \psi = \mathbb{D} \psi,
$$
where the Dirac operator $\mathbb{D}$ is given by
$$
  \mathbb{D} = \gamma^0 \gamma^1 D_x + a(x) \left[ \left( D_\theta + \frac{i \cot \theta}{2} \right) \gamma^0 \gamma^2 + \frac{1}{\sin \theta} D_\varphi \gamma^0 \gamma^3 \right],
$$
and
$$
  D_x = -i \partial_x, \quad D_\theta = -i \partial_\theta, \quad D_\varphi = -i \partial_\varphi.
$$

Let us introduce some notations. We denote
$$
  \Gamma^1 = \gamma^0 \gamma^1, \quad \Gamma^2 = \gamma^0 \gamma^2, \quad \Gamma^3 = \gamma^0 \gamma^3.
$$
From (\ref{Clifford}), it is clear that the Dirac matrices $\Gamma^1, \Gamma^2, \Gamma^3$ satisfy the usual anticommutation relations
\begin{equation} \label{AntiCom}
  \{ \Gamma^i, \, \Gamma^j \} = 2 \delta_{ij}, \quad \forall i,j = 1, 2, 3.
\end{equation}
We choose the following representation for these Dirac matrices
\begin{equation} \label{DiracMatrices}
  \Ga = \left( \begin{array}{cc} 1&0 \\0&-1 \end{array} \right), \quad \Gb = \left( \begin{array}{cc} 0&1 \\1&0 \end{array} \right), \quad \Gc = \left( \begin{array}{cc} 0&i \\-i&0 \end{array} \right).
\end{equation}
We also denote
\begin{equation} \label{DiracSphere}
  \mathbb{D}_{\S^2} = \Gb \left( D_{\theta} + \frac{i \cot{\theta}}{2} \right) + \Gc \frac{1}{\sin{\theta}} D_{\varphi},
\end{equation}
which turns out to be an expression of the intrinsic Dirac operator on $\S^2$. With all these notations, the massless Dirac equation on $(M,\tau)$ takes its final Hamiltonian form
\begin{equation} \label{DiracEq}
  i \partial_t \psi = \mathbb{D} \psi, \quad \mathbb{D} = \Ga D_x + a(x) \mathbb{D}_{\S^2}.
\end{equation}

\begin{remark} \label{Link}
  Consider the Riemanniann manifold $\Sigma = \R_x \times \S^2_{\theta,\varphi}$ equipped with the metric $\sigma = dx^2 + a^{-2}(x) d\omega^2$ where $d\omega^2$ denotes the euclidean metric on $\S^2$. Using the same Cartan's orthonormal frame formalism as described above, we could associate to $(\Sigma,\sigma)$ a Dirac operator $\mathbb{D}_{\sigma}$ and consider the associated Dirac equation
\begin{equation} \label{DiracEq-SV}
  i \partial_t \psi = \mathbb{D}_{\sigma} \psi,
\end{equation}
with $\psi$ given by (\ref{WeightedSpinor}).

The Dirac equation (\ref{DiracEq-SV}) is the one we obtain if we adopt the Schr\"odinger viewpoint, namely if we consider the evolution of Dirac fields on the fixed $3$D-Riemanniann manifold $(\Sigma,\sigma)$. On the other hand, the Dirac equation (\ref{DiracEq}) is the one we obtain if we adopt the relativistic viewpoint, that is the natural Dirac equation associated to the $4$D-Lorentzian manifold $(M = \R_t \times \Sigma, \tau = dt^2 - \sigma)$. It turns out that the two points of view are equivalent in the sense that the corresponding Dirac operators $\mathbb{D}$ and $\mathbb{D}_{\sigma}$ are unitarily equivalent. This can be seen by a direct calculation.

To each point of $(\Sigma,\sigma)$, we associate the orthonormal local frame
$$
  e_1 = \partial_x, \quad e_2 = a(x) \partial_\theta, \quad e_3 = \frac{a(x)}{\sin\theta} \partial_\varphi.
$$
Note that the vector fields $\{e_A\}_{A = 1,2,3}$ satisfy $\sigma(e_A, e_B) = \delta_{AB}$ where $\delta_{AB} =$ diag$(1,1,1)$ is the flat (Riemanniann) $3$D metric. Now following the same procedure as above (still introducing the spinor weight (\ref{WeightedSpinor})), we obtain the following Dirac equation
$$
  i \partial_t \psi = \mathbb{D} \psi, \quad \mathbb{D} = \gamma^1 D_x + a(x) \left[ \gamma^2 \left( D_{\theta} + \frac{i \cot{\theta}}{2} \right) + \gamma^3 \frac{1}{\sin{\theta}} D_{\varphi} \right],
$$
where the gamma Dirac matrices $\gamma^1, \gamma^2, \gamma^3$ satisfy the anticommutation formulae (\ref{AntiCom}). Hence we conclude that the Dirac equations (\ref{DiracEq}) and (\ref{DiracEq-SV}) only differ by a choice of equivalent representation of the gamma Dirac matrices satisfying (\ref{AntiCom}). But it is well known that such two different choices of representation lead to unitarily equivalent Dirac operators (see \cite{Th}).
\end{remark}

%
%

\Section{The stationary scattering} \label{SS}

In this section, we recall  the construction of the stationary representation of the scattering matrix $S(\lambda,n)$ for a fixed
energy $\lambda \in \R$ and all angular momentum $n \in \N$, (we refer to \cite{AKM} and \cite{DN3} for details).
Let us consider first the stationary solutions of equation (\ref{DiracEq}) restricted to each spin weighted spherical harmonic, \textit{i.e.} the solutions of
\begin{equation} \label{PartialSE}
  [ \Ga D_x - n a(x) \Gb ] \psi = \lambda \psi, \quad \forall n \in \N^*.
\end{equation}
For $\lambda \in \R$, we define the Jost solution from the left $F_L(x,\lambda,n)$ and the Jost solution from the right $F_R(x,\lambda,n)$
as the $2\times2$-matrix solutions of (\ref{PartialSE}) satisfying the following asymptotics
\begin{eqnarray}
  F_L(x,\lambda,n) & = & e^{i\Ga \lambda x} (I_2 + o(1)), \ x \to +\infty, \label{FL}\\
  F_R(x,\lambda,n) & = & e^{i\Ga \lambda x} (I_2 + o(1)), \ x \to -\infty. \label{FR}
\end{eqnarray}
From (\ref{PartialSE}), (\ref{FL}) and (\ref{FR}), it is easy to see that such solutions (if there exist) must satisfy the integral equations
\begin{equation} \label{IE-FL}
  F_L(x,\lambda,n) = e^{i\Ga \lambda x} - i n \Ga \int_x^{+\infty} e^{-i\Ga \lambda (y-x)} a(y) \Gb F_L(y,\lambda,n) dy,
\end{equation}
\begin{equation} \label{IE-FR}
  F_R(x,\lambda,n) = e^{i\Ga \lambda x} + i n \Ga \int_{-\infty}^x e^{-i\Ga \lambda (y-x)} a(y) \Gb F_R(y,\lambda,n) dy.
\end{equation}
Since the potential $a$ belongs to $L^1(\R)$, it follows that the integral equations (\ref{IE-FL}) and (\ref{IE-FR}) are uniquely solvable by iteration and that
$$
  \|F_L(x,\lambda,n)\| \leq e^{n \int_x^{+\infty} a(s) ds}, \quad \|F_R(x,\lambda,n)\| \leq e^{n \int_{-\infty}^x a(s) ds}.
$$

Since the Jost solutions are fundamental matrices of (\ref{PartialSE}), there exists a $2\times2$-matrix $A_L(\lambda,n)$ such that $F_L(x,\lambda,n) = F_R(x,\lambda,n) \, A_L(\lambda,n)$. From (\ref{FR}) and (\ref{IE-FL}), we get the following expression for $A_L(\lambda,n)$
\begin{equation} \label{ALRepresentation}
  A_L(\lambda,n) = I_2 - i n \Ga \int_\R e^{-i\Ga \lambda y} a(y) \Gb F_L(y,\lambda,n) dy.
\end{equation}
Moreover, the matrix $A_L(\lambda,n)$ satisfies the following equality (see \cite{AKM}, Proposition 2.2)
\begin{equation} \label{AL-Relation}
  A_L^*(\lambda,n) \Ga A_L(\lambda,n) = \Ga, \quad \forall \lambda \in \R, \ n \in \N.
\end{equation}
Using the notation
\begin{equation} \label{ScatCoef}
  A_L(\lambda,n) = \left[\begin{array}{cc} a_{L1}(\lambda,n)&a_{L2}(\lambda,n)\\a_{L3}(\lambda,n)&a_{L4}(\lambda,n) \end{array} \right],
\end{equation}
the equality (\ref{AL-Relation}) can be written in components as
\begin{equation} \label{ALUnitarity}
  \left. \begin{array}{ccc} |a_{L1}(\lambda,n)|^2 - |a_{L3}(\lambda,n)|^2 & = & 1, \\
  |a_{L4}(\lambda,n)|^2 - |a_{L2}(\lambda,n)|^2 & = & 1, \\
  a_{L1}(\lambda,n) \overline{a_{L2}(\lambda,n)} - a_{L3}(\lambda,n) \overline{a_{L4}(\lambda,n)} & = & 0. \end{array} \right.
\end{equation}
The matrices $A_L(\lambda,n)$ encode all the scattering information of equation (\ref{PartialSE}). In particular, it is shown in \cite{AKM} that the
scattering matrix $S(\lambda,n)$ has the representation
\begin{equation} \label{SR-SM1}
  S(\lambda,n) = \left[ \begin{array}{cc} T(\lambda,n)&R(\lambda,n)\\ L(\lambda,n)&T(\lambda,n) \end{array} \right],
\end{equation}
where
\begin{equation} \label{SR-SM2}
  T(\lambda,n) = \frac{1}{a_{L1}(\lambda,n)}, \quad R(\lambda,n) = - \frac{a_{L2}(\lambda,n)}{a_{L1}(\lambda,n)}, \quad L(\lambda,n) =
  \frac{a_{L3}(\lambda,n)}{a_{L1}(\lambda,n)}.
\end{equation}


\begin{remark}
It follows from (\ref{IE-FL}) that if we define the new potential $\tilde{a}(x)=a(x+c)$, the associated Jost solutions satisfy
\begin{equation}\label{FLdecale}
    \tilde{F_L}(x,\lambda,n) = F_L(x + c,\lambda,n) e^{-i \Ga \lambda c}.
\end{equation}
Hence, it follows from (\ref{ALRepresentation}) that (with obvious notations)
\begin{equation}\label{ALdecale}
  \tilde{A_L}(\lambda,n) = e^{i\Ga \lambda c} A_L(\lambda,n) e^{-i\Ga \lambda c},
\end{equation}
and so, using (\ref{SR-SM1}) and (\ref{SR-SM2}), we conclude that
\begin{equation}
  \tilde{S}(\lambda,n) = e^{i\Ga \lambda c} S(\lambda,n) e^{-i\Ga \lambda c},
\end{equation}
or in components
\begin{equation} \label{Sdecale}
  \left[ \begin{array}{cc} \tilde{T}(\lambda,n)& \tilde{R}(\lambda,n)\\
  \tilde{L}(\lambda,n) & \tilde{T}(\lambda,n) \end{array} \right] = \left[ \begin{array}{cc} T(\lambda,n)& e^{2i \lambda c} R(\lambda,n)\\
  e^{-2i \lambda c} L(\lambda,n)&T(\lambda,n) \end{array} \right].
\end{equation}
Hence the transmission coefficients $T(\lambda,n)$ are invariant under any radial translations of the potential $a$, whereas the reflection coefficients $L(\lambda,n)$ and $R(\lambda,n)$ are invariant under the discrete set of radial translations $\tilde{a}(x) = a(x + \frac{k \pi}{\lambda})$ for $k \in \Z$ and $\lambda \ne 0$.
\end{remark}

Following an original idea due to Regge \cite{Re}, we shall allow the angular momentum $n \in \N$ to take \emph{complex values} $z$ and study the analytic properties of the above scattering data with respect to $z \in \C$. Precisely, it was shown in \cite{DN3} that we can define for $z \in \C$, the Jost solutions $F_L(x,\lambda,z)$ and $F_R(x,\lambda,z)$ which are the unique solutions of the stationary equation
\begin{equation} \label{PartialSE1}
  [ \Ga D_x - z a(x) \Gb ] \psi = \lambda \psi, \quad \forall z \in \C.
\end{equation}
with the asymptotics (\ref{FL}) and (\ref{FR}). Similarly, we can define the matrix $A_L(\lambda,z)$ for all $z \in \C$. All these matrix-functions are analytic in the complex variable $z \in \C$. Moreover, they satisfy the following properties:

\begin{lemma} \label{AL-Analytic}
  (i) Set $A = \displaystyle\int_\R a(x) dx$. Then
  \begin{eqnarray}
    |a_{L1}(\lambda,z)|, \ |a_{L4}(\lambda,z)| \leq \cosh(A|z|), \quad \forall z \in \C, \label{AL-ExpType1}\\
    |a_{L2}(\lambda,z)|, \ |a_{L3}(\lambda,z)| \leq \sinh(A|z|), \quad \forall z \in \C. \label{AL-ExpType2}
  \end{eqnarray}
  (ii) The functions $a_{L1}(\lambda,z)$ and $a_{L4}(\lambda,z)$ are entire and even in $z$ whereas the functions $a_{L2}(\lambda,z)$ and $a_{L3}(\lambda,z)$ are entire and odd in $z$. Moreover they satisfy the symmetries
  \begin{eqnarray}
    a_{L1}(\lambda,z) & = & \overline{a_{L4}(\lambda,\bar{z})}, \quad \forall z \in \C, \label{ALSym1}\\
    a_{L2}(\lambda,z) & = & \overline{a_{L3}(\lambda,\bar{z})}, \quad \forall z \in \C. \label{ALSym2}
  \end{eqnarray}
  (iii) The following relations hold for all $z \in \C$
  \begin{eqnarray}
    a_{L1}(\lambda,z) \overline{a_{L1}(\lambda,\bar{z})} - a_{L3}(\lambda,z) \overline{a_{L3}(\lambda,\bar{z})}
    & = & 1, \label{SymAL1-AL3}\\
    a_{L4}(\lambda,z) \overline{a_{L4}(\lambda,\bar{z})} - a_{L2}(\lambda,z) \overline{a_{L2}(\lambda,\bar{z})}
    & = & 1. \label{SymAL2-AL4}
  \end{eqnarray}
\end{lemma}

At this stage, we have proved that the components of the matrix $A_L(\lambda,z)$
are entire functions of exponential type in the variable $z$.
Precisely, from (\ref{AL-ExpType1}) and (\ref{AL-ExpType2}), we have
\begin{equation} \label{AL-ExpType}
  |a_{Lj}(\lambda,z)| \leq e^{A |z|}, \quad \forall z \in \C, \ j=1,..,4.
\end{equation}
Using the relations (\ref{SymAL1-AL3}), (\ref{SymAL2-AL4}) and the parity properties of the $a_{Lj}(\lambda,z)$,
we can improve these estimates (see Lemma 3.4. in \cite{DN3}).

\begin{lemma} \label{MainEsti}
  Let $\lambda \in \R$ be fixed. Then for all $z \in \C$
  \begin{equation} \label{MainEst}
    |a_{Lj}(\lambda,z)| \leq e^{A |Re(z)|}, \quad j=1,..,4.
  \end{equation}
\end{lemma}
It follows from Lemma \ref{MainEsti} that the functions $z \rightarrow a_{Lj}(\lambda, z)$ belong to the Nevanlinna class in the right half-plane (see for instance \cite{Ru} for a definition). We emphasize that this property is the key point to prove Theorem \ref{globaluniqueness} (see \cite{DN3}).

Similarly, if we use the notation
$$
F_L(x,\lambda,z) = \left[\begin{array}{cc} f_{L1}(x,\lambda,z)&f_{L2}(x,\lambda,z)\\f_{L3}(x,\lambda,z)&f_{L4}(x,\lambda,z) \end{array} \right],
\quad F_R(x,\lambda,z) = \left[\begin{array}{cc} f_{R1}(x,\lambda,z)&f_{R2}(x,\lambda,z)\\f_{R3}(x,\lambda,z)&f_{R4}(x,\lambda,z) \end{array} \right],
$$
we have the corresponding estimates for the Jost functions $f_{Lj}(x,\lambda,z)$ and $f_{Rj}(x,\lambda,z)$ for $j= 1,\dots,4$. Precisely

\begin{lemma} \label{MainEstiF}
  For all $j=1,..,4$ and for all $x \in \R$,
  \begin{eqnarray}
    |f_{Lj}(x,\lambda,z)| \leq C \, e^{|Re(z)| \int_x^{\infty} a(s) ds}, \\
    |f_{Rj}(x,\lambda,z)| \leq C \, e^{|Re(z)| \int_{-\infty}^x a(s) ds}.
  \end{eqnarray}
\end{lemma}

Finally, we shall need later the asymptotic expansion of the scattering data when the angular momentum $z \rightarrow +\infty$, $z$ real. The main tool to obtain these asymptotics easily is a simple change of variable $X=g(x)$, called the Liouville transformation which we precise here. Let us define
\begin{equation} \label{Liouville}
  X= g(x)=\int_{-\infty}^x a(t) \ dt.
\end{equation}
Clearly, since $a > 0$ and continuous, $g:\R \rightarrow ]0,A[$ is a $C^1$-diffeomorphism where
\begin{equation} \label{DefA}
  A=\int_\R a(t) \ dt.
\end{equation}
In what follows, we denote by $h=g^{-1}$ the inverse diffeomorphism of $g$ and we use the notation ${\displaystyle{f'(X) = \frac{\partial f}{\partial X} (X)}}$. We also define for $j=1,...,4$, and for $X \in ]0,A[$,
\begin{equation} \label{Fj}
  \fj = f_{Lj} (h(X), \lambda,z),
\end{equation}
\begin{equation} \label{Gj}
  \gj = f_{Rj} (h(X), \lambda,z).
\end{equation}
Observe at last that in the variable $X$, Lemma \ref{MainEstiF} can be written as
\begin{equation}\label{estimatefg}
  \forall z > 0, \quad  |f_j(X,\lambda,z)| \leq C \, e^{z(A-X)} \ \ ,\ \ |g_j(X,\lambda,z)| \leq C \, e^{z X}.
\end{equation}

The interest in introducing the variable $X$ is that the components $\fj$ and $\gj$ of the Jost solutions satisfy now singular Sturm-Liouville differential equations in the variable $X$, in which the complex angular momentum $z$ plays the role of the spectral parameter. More precisely, we have the following lemma.

\begin{lemma} \label{SturmLiouville} \hfill
\begin{enumerate}
  \item For $j=1,2$, $\fj$ and  $\gj$ satisfy on $]0,A[$ the Sturm-Liouville equation
\begin{equation} \label{SL1}
       y'' +q(X)y = z^2 y.
\end{equation}
  \item For $j=3,4$, $\fj$ and $\gj$ satisfy on $]0,A[$ the Sturm-Liouville equation
    \begin{equation} \label{SL2}
       y'' +\overline{q(X)}y = z^2 y,
    \end{equation}
\end{enumerate}
where the potential $$
{\displaystyle{q(X) = \lambda^2 h'(X)^2 -i \lambda h''(X)= \frac{\lambda^2}{a^2(x)} +i\lambda \frac{a'(x)}{a^3(x)}}},
$$
has the asymptotics
    \begin{eqnarray} \label{omega}
       & &q(X) - \frac{\omega_-}{X^2} = O(1)\ , \ X \rightarrow 0  \ ,\quad \rm{with} \ \ \omega_- = \frac{\lambda^2}{\kappa_-^2} + i \frac{\lambda}{\kappa_-}, \\
       & &q(X) - \frac{\omega_+}{(A-X)^2} = O(1) \ , \ X \rightarrow A \ ,\quad \rm{with} \ \
       \omega_+ = \frac{\lambda^2}{\kappa_+^2} + i \frac{\lambda}{\kappa_+}. \label{omega1}
    \end{eqnarray}
\end{lemma}

A short glance at Lemma \ref{SturmLiouville} suggests that the Jost functions $f_j$ and $g_j$ can be constructed as small perturbations of usual modified Bessel functions $I_{\nu}(z(A-X))$ and $I_{\mu}(zX)$ for suitable $\mu, \ \nu$. This was done in details in \cite{DN3}. As a consequence of this construction and using the well known asymptotic expansion of the modified Bessel functions, the large $z$ asymptotics of the scattering data $a_{Lj}(\lambda,z)$ were calculated in \cite{DN3}. More precisely, if we set
\begin{equation}
\nu_+ = \frac{1}{2} - i\frac{\lambda}{\kappa_+} \ ,\ \mu_- = \frac{1}{2} + i\frac{\lambda}{\kappa_-},
\end{equation}
the following asymptotics hold.

\begin{theorem}\label{asymtoticsutiles} \hfill
\begin{enumerate}
\item For $X \in ]0,A[$ fixed and $z \in S_\theta$ where $S_\theta = \{ z \in \C, \ |\arg(z)| \leq \theta\}$ for a given $0 < \theta < \frac{\pi}{2}$, we have for the Jost solutions $f_1(X)$ and $g_2(X)$
\begin{eqnarray}
 f_1(X,\lambda, z) &=&  \frac{2^{-\nu_+}}{\sqrt{2\pi}}\ (-\frac{\kappa_+}{a_+})^{\frac{i\lambda}{\kappa_+}} \ \Gamma(1-\nu_+)
 \ z^{-\frac{i\lambda}{\kappa_+}} \ e^{z(A-X)} \ \Big(1+O(\frac{1}{z})\Big). \label{asymflun} \\
 g_2(X,\lambda, z) &=&  i \ \frac{2^{-\mu_-}}{\sqrt{2\pi}}\ (\frac{\kappa_-}{a_-})^{-\frac{i\lambda}{\kappa_-}} \ \Gamma(1-\mu_-)
 \ z^{\frac{i\lambda}{\kappa_-}} \ e^{zX} \ \Big(1+O(\frac{1}{z})\Big). \label{asymgdeux}
\end{eqnarray}
\item For the scattering data $a_{L1}(\lambda,z)$ and $a_{L3}(\lambda,z)$, we have
\begin{eqnarray}
\alun  & = &  \frac{1}{2\pi}\ \left(-\frac{\kappa_+}{a_+}\right)^{\frac{i\lambda}{\kappa_+}}
\left(\frac{\kappa_-}{a_-}\right)^{-\frac{i\lambda}{\kappa_-}}
\Gamma\left(\frac{1}{2}-\frac{i\lambda}{\kappa_-}\right) \Gamma\left(\frac{1}{2}+\frac{i\lambda}{\kappa_+}\right)  \nonumber \\
           & & \hspace{1cm} \times \left(\frac{z}{2}\right)^{i\lambda (\frac{1}{\kappa_-} - \frac{1}{\kappa_+})} \ e^{zA} \ \left(1 + O(\frac{1}{z})\right), \label{asymptoticalun} \\
\alt & = & \frac{i}{2\pi}\ \left(-\frac{\kappa_+}{a_+}\right)^{\frac{i\lambda}{\kappa_+}}
\left(\frac{\kappa_-}{a_-}\right)^{\frac{i\lambda}{\kappa_-}}
\Gamma\left(\frac{1}{2}+\frac{i\lambda}{\kappa_-}\right) \Gamma\left(\frac{1}{2}+\frac{i\lambda}{\kappa_+}\right) \nonumber  \\
           & & \hspace{1cm} \times \left(\frac{z}{2}\right)^{-i\lambda (\frac{1}{\kappa_+} + \frac{1}{\kappa_-})}  e^{zA} \ \left(1 + O(\frac{1}{z})\right). \nonumber
\end{eqnarray}
\item For the scattering coefficients $T(\lambda,z)$ and $L(\lambda,z)$, we have
 \begin{eqnarray} \label{asymptoticTL}
T(\lambda,z) & = & 2\pi \frac{\left(-\frac{a_+}{\kappa_+}\right)^{\frac{i\lambda}{\kappa_+}}
\left(\frac{a_-}{\kappa_-}\right)^{-\frac{i\lambda}{\kappa_-}} }{\Gamma\left(\frac{1}{2}-\frac{i\lambda}{\kappa_-}\right) \Gamma\left(\frac{1}{2}+\frac{i\lambda}{\kappa_+}\right)} \ \left(\frac{z}{2}\right)^{i\lambda (\frac{1}{\kappa_+} - \frac{1}{\kappa_-})}  e^{-zA} \left(1 + O(\frac{1}{z})\right), \nonumber \\
L(\lambda,z) & = & i\ \left(\frac{\kappa_-}{a_-}\right)^{\frac{2i\lambda}{\kappa_-}}\
\frac{\Gamma\left(\frac{1}{2}+\frac{i\lambda}{\kappa_-}\right)}{\Gamma\left(\frac{1}{2}-\frac{i\lambda}{\kappa_-}\right)}\
\left(\frac{z}{2}\right)^{-\frac{2i\lambda}{\kappa_-}} \left(1 + O(\frac{1}{z})\right).
\end{eqnarray}
\end{enumerate}
\end{theorem}

%
%

\Section{Proofs of the local inverse scattering results} \label{LIS}

\subsection {Proof of Theorem \ref{Reflection}, $(i) \Rightarrow (ii)$.}

Assume that $L(\lambda,n) = \tilde{L}(\lambda,n) + \ O\left(e^{-2nB} \right)$, $n \rightarrow + \infty.$ Our first step is to extend these asymptotics
(which are true for $n$ integer  $\rightarrow + \infty$) to the case of $z \rightarrow +\infty$ ($z$ real and positive). To do this, we shall use some well-known uniqueness results for Laplace transforms obtained in \cite{Ho1, Si}.

We begin with an elementary result for functions of the complex variable belonging to the Hardy class. We recall (see for instance \cite{Lev}, Lecture 19) that the Hardy class $H_+^2$ is the set of analytic functions $F$ in the right half-plane
$\Omega = \{ z \in \C \ ,\ Re \ z >0 \}$, satisfying the condition
\begin{equation}\label{Hardy}
\sup_{x >0} \ \int_{\R} \ \mid F(x+i y) \mid ^2 \ dy \ < \infty,
\end{equation}
and equipped with the norm
\begin{equation}\label{norme}
\mid \mid F \mid \mid = \left( \sup_{x >0} \ \int_{\R} \ \mid F(x+i y) \mid ^2 \ dy  \right)^{\frac{1}{2}}.
\end{equation}
The Paley-Wiener Theorem asserts that a function $F(z)$ belongs to the Hardy space $H_+^2$ if and only if there exists a function $f \in L^2(0,+\infty)$ such that
\begin{equation}\label{PaleyWiener}
F(z) = \frac{1}{\sqrt{2\pi}}\ \int_0^{+\infty} e^{-tz} \ f(t) \ dt \ ,\ \forall z \in \Omega.
\end{equation}
Moreover, we have 
\begin{equation} \label{normeL2}
\mid \mid F \mid \mid \  = \   \mid \mid f \mid \mid_{L^2(0,\infty)}.
\end{equation}

Let us also recall a uniqueness result for Laplace transforms given in \cite{Ho1}, Prop. 2.4., (see also \cite{Si} for a continuous version):

\begin{prop}\label{uniciteLaplace}
Let $f \in L^1(0,a)$. If for all $\epsilon >0$,
$$
\int_0^a \ e^{-nt}  f(t)  \ dt \ =\ O(e^{-an(1-\epsilon)})\ , \ \ n \rightarrow +\infty,
$$
then $f=0$ a.e.
\end{prop}

We now put together all the previous results and prove the following Proposition.

\begin{prop}\label{PropHardy}
Let $F$ be a function in the Hardy class $H_+^2 $. Assume that for some $B>0$, we have $F(n) = O\left( e^{-Bn} \right), \ n \rightarrow +\infty$, ($n$ integer).
Then,
\begin{equation} \label{majoration}
\mid F(z) \mid  \ \leq \  \frac{\mid \mid F \mid \mid}{\sqrt{4\pi Re z}}  \ e^{-B  Re  z} \ , \ \forall z \in \Omega.
\end{equation}
\end{prop}

\begin{proof}
For $n\in \N$, the Paley-Wiener theorem and the Cauchy-Schwarz inequality imply
\begin{eqnarray*}
\int_0^B e^{-nt} f(t) \ dt &=& \sqrt{2\pi} \  F(n) - \int_B^{+\infty} e^{-nt} f(t) \ dt  \\
                           &=& \sqrt{2\pi} \  F(n) +O\left(e^{-nB} \right) \ = \ O\left(e^{-nB} \right).
\end{eqnarray*}
So, Proposition \ref{uniciteLaplace} entails that $f = 0$ a.e in $(0,B)$. Using (\ref{PaleyWiener}) again and (\ref{normeL2}), we obtain at once (\ref{majoration}).
\end{proof}

Let us give a direct consequence (which could be certainly improved) of the previous result to our inverse problem.

\begin{prop}\label{estimL}
Assume that the reflection coefficients $L(\lambda, n)$ and $\tilde{L}(\lambda, n)$ satisfy for some $0 < B < \min(A,\tilde{A})$,
$$
L(\lambda,n) =  \tilde{L}(\lambda,n) + O(e^{-2nB}) \ , \ n \rightarrow + \infty, \ \ n \ \textrm{integer}.
$$
Then
\begin{equation}
L(\lambda,z) =  \tilde{L}(\lambda,z) + O(\sqrt{z} \  e^{-2zB}) \ , \ z \rightarrow + \infty, \ \ z \ \textrm{real}.
\end{equation}
\end{prop}

\begin{proof}
First, let us recall that
\begin{equation}
L(\lambda, n) = \frac{a_{L3}(\lambda, n)}{a_{L1}(\lambda,n)} \ \ ,\ \  \tilde{L}(\lambda, n) = \frac{\tilde{a}_{L3}(\lambda, n)}{\tilde{a}_{L1}(\lambda,n)}.
\end{equation}
Using Lemma \ref{MainEsti}, we obtain immediately
\begin{equation}\label{estim}
a_{L3}(\lambda, n) \tilde{a}_{L1}(\lambda,n) - \tilde{a}_{L3}(\lambda, n) a_{L1}(\lambda, n)  = O(e^{n(A+\tilde{A}-2B)}).
\end{equation}
For $z \in \Omega$, we set
\begin{equation}\label{fonctionF}
F(z) = \frac{ a_{L3}(\lambda, z) \tilde{a}_{L1}(\lambda,z) - \tilde{a}_{L3}(\lambda, z) a_{L1}(\lambda, z)} {z+1} \ e^{-z(A+\tilde{A})}.
\end{equation}
Clearly, $F$ is holomorphic in $\Omega$, and by Lemma \ref{MainEsti}, we have
\begin{equation}
\mid F(z) \mid \ \leq \ \frac{2}{\mid z+1 \mid}.
\end{equation}
It follows that $F \in H_+^2$ and by (\ref{estim}), we have $F(n) = O(e^{-2nB})$. Using Proposition \ref{PropHardy}, we see that
$F(z) = O( z^{-\frac{1}{2}} e^{-2zB})$, $z \rightarrow +\infty$. For $z>0$, we write
\begin{equation}
L(\lambda, z) - \tilde{L}(\lambda,z) =  \frac{(z+1)\ e^{z(A+\tilde{A})}} {a_{L1}(\lambda, z) \ \tilde{a}_{L1}(\lambda,z)} \ F(z).
\end{equation}
We conclude the proof using (\ref{asymptoticalun}).
\end{proof}

This concludes the first step of the proof of Theorem \ref{Reflection}, $(i) \Rightarrow (ii)$. The second step of the proof consists in an adaption of the strategy used to prove the local Borg-Marchenko Theorem in \cite{Si, Be} for one-dimensional Schr\"odinger operators to our setting of Dirac operators on a SSAHM. This strategy is relatively close to the proof of Theorem 1.1 given in \cite{DN3}, itself inspired by \cite{FY}. Let us introduce for $X \in ]0,B[$ the matrix
$$
  P(X,\lambda,z) = \left( \begin{array}{cc} P_1(X,\lambda,z) & P_2(X,\lambda,z) \\
                                            P_3(X,\lambda,z) & P_4(X,\lambda,z)
                                            \end{array} \right),
$$
defined by
\begin{equation}\label{matricepassage}
P(X,\lambda,z) \ \tilde{F}_R (\tilde{h} (X), \lambda, z) \ = \ F_R (h(X), \lambda, z),
\end{equation}
where $F_R =(f_{Rk})$ and $\tilde{F}_R =(\tilde{f}_{Rk})$ are the Jost solutions from the right associated with $a(x)$ and $\tilde{a}(x)$.
To simplify the notations, for $k=1, ...,4$, we set as previously:
\begin{eqnarray*}
f_{k} (X,\lambda, z) = f_{Lk} (h (X), \lambda, z), & \tilde{f}_{k} (X,\lambda, z) = \tilde{f}_{Lk} (\tilde{h} (X), \lambda, z),\\
g_{k} (X,\lambda, z) = f_{Rk} (h (X), \lambda, z), & \tilde{g}_{k} (X,\lambda, z) = \tilde{f}_{Rk} (\tilde{h} (X), \lambda, z).
\end{eqnarray*}
Using that det $F_R =1$ and det $\tilde{F}_R = 1$, we obtain the following equalities :
\begin{equation}\label{premiereformule}
\left\{ \begin{array}{ccc}
P_1(X,\lambda, z) &=&  \ g_{1}\  \tilde{g}_{4} -  \ g_{2} \ \tilde{g}_{3},  \\
P_2(X,\lambda, z) &=& -  \ g_{1} \ \tilde{g}_{2} +  \ g_{2} \ \tilde{g}_{1} .
\end{array}
\right.
\end{equation}
It follows from (\ref{premiereformule}) and the analytical properties of the Jost functions that, for $j=1,2$,
the applications $z \rightarrow P_j(X,\lambda,z)$ are analytic on $\C$ and of exponential type.
Moreover, by Lemma \ref{MainEstiF}, these applications are bounded on the imaginary axis $i \R$.

We shall now prove that the applications $z \rightarrow P_j(X,\lambda,z)$ are also bounded on the real axis.
To do this, we first perform some elementary algebraic transformations on $P_j(X,\lambda,z)$.
Since $F_L(x,\lambda,z)=F_R (x, \lambda, z)\ A_L (\lambda, z)$, we easily get for $z>0$,
\begin{eqnarray}
g_1 &=& \frac{f_1}{a_{L1}} - L(\lambda, z) g_2, \\
\tilde{g}_3 &=& \frac{\tilde{f}_3}{\tilde{a}_{L1}} - \tilde{L}(\lambda, z) \tilde{g}_4 .
\end{eqnarray}
Thus,
\begin{equation}
P_1(X,\lambda, z) = (\tilde{L}(\lambda,z) -L(\lambda,z)) \ g_2 \tilde{g}_4 + \left( \frac{f_1 \tilde{g}_4}{a_{L1}} - \frac{ \tilde{f}_3 g_2}{\tilde{a}_{L1}} \right).
\end{equation}
Using (\ref{estimatefg}) and (\ref{asymptoticalun}), it is easy to see that the function
{$\displaystyle{ z \rightarrow \left( \frac{f_1 \tilde{g}_4}{a_{L1}} - \frac{ \tilde{f}_3 g_2}{\tilde{a}_{L1}} \right)}$}
is  bounded on $\R^+$ for all fixed $X\in ]0,A[$. Moreover, (\ref{estimatefg}) and Proposition \ref{estimL} imply
\begin{equation}
\mid (\tilde{L}(\lambda,z) -L(\lambda,z))\ g_2 \tilde{g}_4 \mid \ \leq \ C \sqrt{z} e^{-2z (B-X)},
\end{equation}
and thus, this term remains bounded when $z \to +\infty$ for all $X \in ]0,B[$. Summarizing, for all fixed $X \in ]0,B[$, the function $ z \rightarrow P_1 (X,\lambda, z)$ is bounded on $\R^+$.

Similarly, we have
\begin{equation}
P_2(X,\lambda, z) = (\tilde{L}(\lambda,z) -L(\lambda,z)) \ g_2 \tilde{g}_2 + \left(\frac{ \tilde{f}_1 g_2}{\tilde{a}_{L1}} - \frac{f_1 \tilde{g}_2}{a_{L1}}\right),
\end{equation}
and using the same arguments as above, we obtain that, for all fixed $X \in ]0,B[$, $ z \rightarrow P_2 (X,\lambda, z)$ is bounded on $\R^+$.

Clearly, these last results remain true on $\R$ by an elementary parity argument. Finally, applying the Phragmen-Lindel\"of's Theorem (\cite{Bo}, Thm 1.4.2.) on each quadrant of the complex plane, we deduce that $z \rightarrow P_j (X,\lambda,z)$ is bounded on $\C$. By Liouville's Theorem, and a standard continuity argument in the variable $X$, we have thus obtained
\begin{equation}\label{thliouville}
P_j (X,\lambda,z)=P_j (X,\lambda,0) \ \ ,\ \ \forall z \in \C \ ,\ \forall X \in ]0,B].
\end{equation}

Now, we return to the definition of $P_j(X,\lambda,z)$ for $z=0$. We observe first that ${\displaystyle{F_R (x,\lambda, 0) = e^{i \lambda \Gamma^1 x} }} $
and similarly ${\displaystyle{\tilde{F}_R (x,\lambda, 0) = e^{i \lambda \Gamma^1 x} }} $. This is immediate from the definition of the Jost functions. Thus we deduce from (\ref{matricepassage}) that
\begin{equation}\label{egalitemat}
P(X,\lambda, 0) = e^{i  \lambda \ ( h(X)-\tilde{h}(X))\ \Gamma^1 }.
\end{equation}
Then, putting (\ref{egalitemat}) and (\ref{thliouville}) into (\ref{matricepassage}) we get
\begin{equation}\label{egalitejost}
\left\{ \begin{array}{ccc}
\tilde{g}_{1}(X,\lambda, z) &=& e^{i \lambda \ (\tilde{h}(X) - h (X))}  \ g_{1}(X,\lambda, z),   \\
\tilde{g}_{2}(X,\lambda, z) &=& e^{i \lambda \ (\tilde{h}(X) - h (X))} \ g_{2}(X,\lambda, z).
\end{array}
\right.
\end{equation}

By Lemma 4.2 in \cite{DN3}, the Wronskians $W(g_{1} , g_{2}) = W(\tilde{g}_{1}  , \tilde{g}_{2}) = iz$. Then, a straightforward calculation gives
\begin{equation}\label{egalitephase}
e^{2i \lambda \ (\tilde{h}(X) - h (X))} \ =\ 1.
\end{equation}
Thus, by a standard continuity argument, there exists $k \in \Z$ such that
\begin{equation}\label{diffeos}
\tilde{h}(X)= h(X) + \frac{k \pi}{\lambda} \ \ ,\ \ \forall X \in ]0,B].
\end{equation}
Note that, for the particular choice $X=B$, we obtain ${\displaystyle{\tilde{h}(B)= h(B) + \frac{k \pi}{\lambda}}}$.
Let us differentiate (\ref{diffeos}) with respect to $X$. We obtain easily
\begin{equation}
\frac{1}{a(\tilde{h}(X))} = \frac{1}{a(h(X))},
\end{equation}
and using again (\ref{diffeos}), we have
\begin{equation} \label{UnicitePot}
  a(x) = \tilde{a}(x + \frac{k \pi}{\lambda}) \ \ , \ \  \forall x \in ]-\infty,h(B)].
\end{equation}
Thus, we have proved the first part of Theorem \ref{Reflection}. $\Box$


\subsection {Proof of Theorem \ref{Reflection}, $(ii) \Rightarrow (i)$.}

Let us assume there exists $k \in \Z$ such that $a(x)= \tilde{a}(x + \frac{k \pi}{\lambda}), \ \forall x \leq h(B)$. It follows immediately from the definition of the diffeomorphisms $h$ and $\tilde{h}$, that ${\displaystyle{\tilde{h}(B)= h(B) + \frac{k \pi}{\lambda}}}$. Moreover, if we set $\breve{a}(x) = \tilde{a}(x + \frac{k \pi}{\lambda}), \ \forall x \in \R$, and using (\ref{Sdecale}), we see that (with obvious notation),
\begin{equation}
\breve{L}(\lambda,n) = e^{-2i\lambda \ \frac{k \pi}{\lambda}} \ \tilde{L}(\lambda, n) = \tilde{L}(\lambda, n).
\end{equation}
Thus, it remains to prove the implication $(ii) \Rightarrow (i)$ in the case $k=0$. Now, let us begin with an obvious lemma (whose proof is omitted) :

\begin{lemma}
Assume that $a(x) = \tilde{a}(x), \ \forall x \leq h(B)=\tilde{h}(B)$. Then,
\begin{equation}\label{uniciteparam}
a_- = \tilde{a}_- \ ,\ \kappa_- = \tilde{\kappa}_-.
\end{equation}
and
\begin{equation}\label{uniciteg}
g_j(X,\lambda,z) = \tilde{g}_j(X,\lambda, z), \ \forall X \leq B, \  \forall j=1, \ldots 4.
\end{equation}
\end{lemma}

\noindent Using again the relation $F_L(x,\lambda,z) = F_R(x,\lambda,z) A_L(\lambda, z)$ and $(\ref{uniciteg})$, we have for $z>0$ and $X \leq B$,
\begin{equation}
 \frac{f_1}{a_{L1}} \ =\ g_1 +L(\lambda,z)\  g_2 \ ,\ \frac{\tilde{f}_1}{\tilde{a}_{L1}} \ =\ g_1 +\tilde{L}(\lambda,z)\  g_2.
\end{equation}
For $z>0$ large enough,  (\ref{asymgdeux}) implies that $g_2 \not=0$. So, for such $z$, we can write
\begin{equation}
L(\lambda,z) - \tilde{L}(\lambda,z) \ =\ \frac{1}{g_2} \left( \frac{f_1}{a_{L1}} - \frac{\tilde{f}_1}{\tilde{a}_{L1}} \right).
\end{equation}
Now, using Theorem \ref{asymtoticsutiles}, we obtain easily:
\begin{equation}
L(\lambda,z) - \tilde{L}(\lambda,z) \ =\ O \left( e^{-2zX} \right) \ ,\ \forall X \in ]0, B],
\end{equation}
and taking $X=B$, the proof is complete. $\Box$


\subsection {Proof of Theorem \ref{Reflection}, $(iii) \Leftrightarrow (iv)$.}

The local uniqueness result for the reflection coefficient $R(\lambda,n)$ is actually a by-product of the previous one
using the following trick. If we  set $a^{\star}(x) = a(-x)$, a straightforward calculation using (\ref{IE-FL}) - (\ref{IE-FR}) shows
that the associated Jost solutions satisfy

\begin{equation}\label{egalitestar}
\left\{ \begin{array}{ccc}
F_{R}^{\star}(x,\lambda,n) &=& F_{L}(-x,-\lambda,-n),   \\
F_{L}^{\star}(x,\lambda,n) &=& F_{R}(-x,-\lambda,-n).   \\
\end{array}
\right.
\end{equation}
It follows immediately that $A_{L}^{\star}(\lambda,n) = A_{L}^{-1}(-\lambda,-n)$ which implies the equality $R^{\star}(\lambda,n)= - \overline{L(-\lambda,n)}$.
Thus, it suffices to use the previous result for the reflection coefficients $L$, with $\lambda$ replaced by $-\lambda$, to prove the equivalence $(iii) \Leftrightarrow (iv)$ of Theorem \ref{Reflection}. $\Box$


\subsection {Proof of Theorem \ref{Transmission}.}

Assume that
\begin{equation}
T(\lambda,n) = \tilde{T}(\lambda,n) + \ O\left(e^{-2nB} \right),
\end{equation}
with $B > \max(A,A')$. Using the asymptotics in Theorem \ref{asymtoticsutiles}, we obtain $A=\tilde{A}$ and
\begin{equation}
a_{L1}(\lambda,n) - \tilde{a}_{L1}(\lambda,n) = O\left(e^{-2n(B-A)} \right).
\end{equation}
Now, we set for $z \in \Omega$,
\begin{equation}
F(z) \ =\ \frac{ a_{L1}(\lambda,z) - \tilde{a}_{L1}(\lambda,z)}{z+1} \ e^{-zA}.
\end{equation}
As previously, we see that $F$ belongs to the Hardy space $H_+^2$ and $F(n) = O\left(e^{-(2B-A)n} \right)$. By Proposition \ref{PropHardy}, we have
\begin{equation} \label{majorationF}
\mid F(z) \mid  \ \leq \ \frac{\mid \mid F \mid \mid}{\sqrt{4 \pi Re z}}  \ e^{-(2B-A)  Re  z} \ , \ \forall z \in \Omega.
\end{equation}
It follows that there exists $C>0$ such that for all $z>0$,
\begin{equation}
\mid a_{L1}(\lambda,z) - \tilde{a}_{L1}(\lambda,z) \mid \ \leq \ C \sqrt{z} \ e^{-2z(B-A)}.
\end{equation}
Thus, $f(z) := a_{L1}(\lambda,z) - \tilde{a}_{L1}(\lambda,z)$ is bounded on $\R^+$. Moreover, this function is of exponential type, and bounded on $i\R$. The Phragmen - Lindel\"of Theorem implies that $f$ is bounded on $\Omega$ and consequently, is also bounded on $\C$ using parity arguments. Hence Liouville's Theorem entails that $f(z) = f(0) = 0$. We conclude the proof using Proposition \ref{unicitetransmission}.  $\Box$.


\Section{Inverse uniqueness results in Reissner-Nordstr\"om-de-Sitter black holes} \label{BH}

In this Section, we adapt the previous local inverse uniqueness results to the setting of general relativity and more precisely to Reissner-Nordstr\"om-de-Sitter black holes. We emphasize that the link between such black holes and SSAHM was already given in \cite{DN3}. Considering the scattering of massless Dirac fields evolving in the outer region of a RN-dS black holes, we shall prove that the partial knowledge of the corresponding reflection coefficients in the sense of (\ref{Ln}) - (\ref{Rn}) not only determines the metric of such black holes in the neighbourhood of the event and cosmological horizons (see below for the definition), but in fact determines the whole metric. This is due to the fact that the metric of RN-dS black holes only depend on $3$ parameters - their mass, electric charge and positive cosmogical constant - parameters that can be deduced from the explicit form of the metric in the neighbourhoods of the horizons.

\subsection{Reissner-Nordst\"om-de-Sitter black holes}

Refering to Wald \cite{W} for more general details on black hole spacetimes, we summarize here the essential features of Reissner-Nordstr\"om-de-Sitter (RN-dS) black holes given in \cite{DN2,DN3}. First, RN-dS are spherically symmetric electrically charged exact solutions of
the Einstein-Maxwell equations. In Schwarzschild coordinates, the exterior region of a RN-dS black hole is described by
the four-dimensional manifold $\M = \mathbb{R}_{t} \times ]r_-,r_+[_r \times \S_{\theta,\varphi}^{2}$ equipped with the
Lorentzian metric
\begin{equation} \label{Metric}
  \tau = F(r)\,dt^{2} - F(r)^{-1} dr^{2} - r^{2} \big(d\theta^{2}+\sin^{2}\theta \, d\varphi^{2}\big),
\end{equation}
where
\begin{equation} \label{F}
  F(r) = 1-\frac{2M}{r}+\frac{Q^{2}}{r^{2}} - \frac{\Lambda}{3} r^2.
\end{equation}
The constants $M>0$, $Q \in \R$ appearing in (\ref{F}) are interpreted as the mass and the electric charge of the black
hole and $\Lambda > 0$ is the cosmological constant of the universe. We assume here that the function $F(r)$ has three
simple positive roots $0<r_c<r_-<r_+$ and a negative one $r_n <0$. This is always achieved if we suppose for instance
that $Q^2<\frac{9}{8} M^2$ and that $\Lambda M^2$ be small enough (see \cite{L}). The sphere $\{r=r_c\}$ is called the
Cauchy horizon whereas the spheres $\{r=r_-\}$ and $\{r=r_+\}$ are the event and cosmological horizons respectively. These horizons which appear as singularities of the metric (\ref{Metric}) are in fact mere coordinates singularities. This means that using appropriate coordinates system, these horizons can be understood as regular null hypersurfaces that can be crossed one way but would require speeds greater than that of light to be crossed the other way: hence their names horizons.

In what follows, we shall only consider the exterior region of the black hole, that is the region $\{r_- < r <r_+\}$ lying between the event and cosmological horizons. Note that the function $F$ is positive there. The point of view implicitly adopted here is indeed that of static observers located far from the event and cosmological horizons of the black hole. We think typically of a telescope on earth aiming at the black hole or at the
cosmological horizon. We understand these observers as living on worldlines $\{r = r_0\}$ with $r_- << r_0 << r_+$. The variable $t$ corresponds to their true perception of time. From the point of view of our static observers, the event and cosmological horizons turn out to be the boundaries of the \emph{observable} world. This can be more easily understood if we remark that the event and cosmological horizons are never reached in a finite time $t$ by incoming and outgoing radial null geodesics, the trajectories followed by classical light-rays aimed radially at the black hole or at the cosmological horizon. Both horizons are thus perceived as \emph{asymptotic regions} by our static observers.

Instead of working with the radial variable $r$, we make the choice to describe the exterior region of the black hole
by using the Regge-Wheeler (RW) radial variable which is more natural when studying the scattering properties of any
fields. The RW variable $x$ is defined implicitly by $\frac{dx}{dr} = F^{-1}(r)$, or explicitly by
\begin{equation} \label{RW}
  x = \frac{1}{2\kappa_n} \ln(r-r_n) + \frac{1}{2\kappa_c} \ln(r-r_c) +\frac{1}{2\kappa_-} \ln(r-r_-) + \frac{1}{2\kappa_+} \ln(r_+-r) \ + \ c,
\end{equation}
where $c$ is any constant of integration and the quantities $\kappa_j, \ j=n,c,-,+$ are defined by
\begin{equation} \label{SurfaceGravity}
  \kappa_n = \frac{1}{2} F'(r_n), \ \kappa_c = \frac{1}{2} F'(r_c), \ \kappa_- = \frac{1}{2} F'(r_-), \ \kappa_+ = \frac{1}{2} F'(r_+).
\end{equation}
The constants $\kappa_->0$ and $\kappa_+<0$ are called the surface gravities of the event and cosmological horizons respectively.
Note from (\ref{RW}) that the event and cosmological horizons $\{r=r_\pm\}$  are pushed away to the infinities $\{x = \pm \infty\}$ using the RW variable.
Let us also emphasize that the incoming and outgoing null radial geodesics become straight lines $\{x=\pm t\}$ in this new coordinates system, a fact that provides a
natural manner to define the scattering data simply by mimicking the usual definitions in Minkowski-spacetime.

\subsection{The Dirac equation in RN-dS}

As waves, we consider massless Dirac fields propagating in the exterior region of a RN-dS black hole. We refer to \cite{Me1,Ni} for a detailed study of this equation in this background including a complete time-dependent scattering theory. We shall use the expression of the equation obtained in these papers as the starting point of our study. Thus the considered massless Dirac fields are represented by 2 components spinors $\psi$ belonging to the Hilbert space $L^2(\R \times \S^2; \, \C^2)$ which satisfy the evolution equation
\begin{equation} \label{FullEq}
  i \partial_{t} \psi = \Big(\Ga D_x + a(x) D_{\S^2} \Big) \psi,
\end{equation}
The symbol $D_x$ stands for $-i\d_x$ whereas $D_{\S^2}$ denotes the Dirac operator on $\S^2$ which, in spherical coordinates, takes the form (\ref{DiracSphere}). The potential $a$ is the scalar smooth function given in term of the metric (\ref{Metric})-(\ref{F}) by
\begin{equation} \label{Potential}
  a(x) = \frac{\sqrt{F(r(x))}}{r(x)},
\end{equation}
where $r(x)$ is the inverse diffeomorphism of (\ref{RW}). It was shown in \cite{DN3} that the potential $a$ verifies the hypotheses (\ref{AsympA}). Finally, the matrices $\Ga, \Gb, \Gc$ appearing in (\ref{FullEq}) and (\ref{DiracSphere}) are usual $2 \times 2$ Dirac matrices that satisfy the anticommutation relations (\ref{AntiCom}). As before, we shall work with the following representations of the Dirac matrices
$$ 
  \Ga = \left( \begin{array}{cc} 1&0 \\0&-1 \end{array} \right), \quad \Gb = \left( \begin{array}{cc} 0&1 \\1&0 \end{array} \right), \quad \Gc = \left( \begin{array}{cc} 0&i \\-i&0 \end{array} \right).
$$

Hence, the massless Dirac equation on the exterior region of a RN-dS black hole can be put under the same form as the massless Dirac equation on a SSHAM studied in the previous Sections. We can thus define the transmission coefficients $T(\lambda,n)$ and reflection coefficients $L(\lambda,n)$ and $R(\lambda,n)$ for a fixed energy $\lambda \in \R$ and all angular momenta $n \in \N^*$ as before. Moreover, Theorem \ref{Reflection}  remains true in this new setting. Taking advantage of the particular form of the potential $a$ given in (\ref{Potential}), we can slightly improve these results.


\subsection{Uniqueness of the parameters}

\begin{theorem}
 Using the notations of the Theorem \ref{Reflection}, the following assertions are equivalents :
$$
  (i) \quad L(\lambda,n) = \tilde{L}(\lambda,n) + \ O\left(e^{-2nB} \right).
$$
$$
  (ii) \quad R(\lambda,n) = \tilde{R}(\lambda,n) + \ O\left(e^{-2nB} \right).
$$
$$
  (iii) \quad M=\tilde{M}, \quad Q^{2} = \tilde{Q}^{2} \quad \mathrm{and} \quad \Lambda = \tilde{\Lambda}.
$$
$$
  (iv) \quad \exists k \in \Z, \quad a(x) = \tilde{a}(x + \frac{k\pi}{\lambda}), \quad \forall \ x  \in \mathbb{R}.
$$
\end{theorem}

\begin{proof}
  We first use Theorem \ref{Reflection} to obtain the equality of the potential $a$ and $\tilde{a}$ on a half-line $]- \infty,b_1]$ (respectively $[b_2,+\infty[$), $b_1 < b_2 \in \mathbb{R}$. Then, a line by line inspection of the proof given in \cite{DN3} p. 43-44 shows that this information is enough to prove the uniqueness of the mass $M$, the square of the charge $Q^2$ and the cosmological constant $\Lambda$ of the black hole. Finally, since the parameters of the black hole determine uniquely the metric, we obtain the equality of the
potentials $a$ and $\tilde{a}$ on $\mathbb{R}$ (up to a discrete set of translations as stated in (iv)).
\end{proof}

\appendix

\Section{Other formulation of the main inverse uniqueness results}

In this Section, we formulate our main Theorems \ref{Reflection}  in a more global way, avoiding the use of a decomposition onto generalized spherical harmonics. More precisely, we replace the main assumptions (\ref{Ln}) and (\ref{Rn}) by $L^2(\S^2)$-operator norms conditions on the global reflection coefficients. We recall first the definition and essential properties of these operators.

\begin{prop} \label{TkRkLk}
For all $(n,k) \in I$ where $I = \{ n \in \N^*, \ k \in 1/2 + \Z, \ |k| \leq n - \frac{1}{2}\}$, we use the notation $Y_{kn} = (Y_{kn}^1, Y_{kn}^2)$ for the corresponding generalized spherical harmonics. Then, \\

\noindent 1) The families $\{Y_{kn}^1\}_{(n,k) \in I}$ and $\{Y_{kn}^2\}_{(n,k) \in I}$ form Hilbert bases of $\mathfrak{l} = L^2(\S^2; \C)$; precisely for all $\psi \in \mathfrak{l}$, we can decompose $\psi$ as
$$
  \psi = \sum_{n,k \in I} \psi_{kn}^{j} Y_{kn}^{j}, \quad j=1,2,
$$
with
$$
  \| \psi \|^2 = \frac{1}{2} \sum_{n,k \in I} |\psi_{kn}^{j} |^2.
$$

\noindent 2) Let $\lambda \in \R$ be a fixed energy. Then, the transmission operators $T_L(\lambda)$ and $T_R(\lambda)$ are defined as operators from $\mathfrak{l}$ to $\mathfrak{l}$ as follows. For all $\psi = \sum_{n,k \in I} \psi_{kn}^{j} Y_{kn}^{j} \in \mathfrak{l}$
\begin{equation} \label{TL}
  T_L(\lambda) \psi = T_L(\lambda) \left( \sum_{n,k \in I} \psi_{kn}^1 Y_{kn}^1(\lambda) \right) = \sum_{n,k \in I} \left( T(\lambda,n) \psi_{kn}^1 \right) Y_{kn}^1(\lambda),
\end{equation}
and
\begin{equation} \label{TR}
  T_R(\lambda) \psi = T_R(\lambda) \left( \sum_{n,k \in I} \psi_{kn}^2 Y_{kn}^2(\lambda) \right) = \sum_{n,k \in I} \left( T(\lambda,n) \psi_{kn}^2 \right) Y_{kn}^2(\lambda),
\end{equation}
where $T(\lambda,n)$ are the transmission coefficients defined in (\ref{SR-SM1}) - (\ref{SR-SM2}).
In short, we write
\begin{equation} \label{TLR}
  T_L(\lambda) Y_{kn}^1 = T(\lambda,n) Y_{kn}^1, \quad \quad T_R(\lambda) Y_{kn}^2 = T(\lambda,n) Y_{kn}^2, \quad \forall n,k \in I,
\end{equation}
and thus the operators $T_L(\lambda)$ (resp. $T_R(\lambda)$) are diagonalizable on the Hilbert basis of eigenfunctions $(Y_{kn}^1)_{k,n \in I}$ (resp. $(Y_{kn}^2)_{k,n \in I}$) associated to the eigenvalues $T(\lambda,n)$ (in both cases). \\

\noindent 3) Let $\lambda \in \R$ be a fixed energy. Then, the reflection operators $L(\lambda)$ and $R(\lambda)$ are defined as operators from $\mathfrak{l}$ to $\mathfrak{l}$ as follows. For all $\psi = \sum_{n,k \in I} \psi_{kn}^{j} Y_{kn}^{j} \in \mathfrak{l}$
\begin{equation} \label{R}
  R(\lambda) \psi = R(\lambda) \left( \sum_{n,k \in I} \psi_{kn}^2 Y_{kn}^2(\lambda) \right) = \sum_{n,k \in I} \left( R(\lambda,n) \psi_{kn}^2 \right) Y_{kn}^1(\lambda),
\end{equation}
and
\begin{equation} \label{L}
  L(\lambda) \psi = L(\lambda) \left( \sum_{n,k \in I} \psi_{kn}^1 Y_{kn}^1(\lambda) \right) = \sum_{n,k \in I} \left( L(\lambda,n) \psi_{kn}^1 \right) Y_{kn}^2(\lambda),
\end{equation}
where $R(\lambda,n)$ and $L(\lambda,n)$ are defined in (\ref{SR-SM1}) - (\ref{SR-SM2}). In short, we write
\begin{equation} \label{Rek}
  R(\lambda) Y_{kn}^2 = R(\lambda,n) Y_{kn}^1, \quad \quad L(\lambda) Y_{kn}^1 = L(\lambda,n) Y_{kn}^2, \quad \forall n,k \in I.
\end{equation}
\end{prop}

It is immediate from the above definitions to express the $L^2(\S^2)$-operator norms of the transmission operators $T_L(\lambda), T_R(\lambda)$ and of the reflection operators $L(\lambda), R(\lambda)$ in terms of the coefficients $T(\lambda,n), L(\lambda,n)$ and $R(\lambda,n)$. For a fixed $\lambda \in \R$, we have
\begin{equation} \label{LinkTLR}
  \| T_L(\lambda)\| = \|T_R(\lambda)\| = \| T(\lambda,n) \|_\infty, \quad \|L(\lambda)\| = \| L(\lambda,n) \|_\infty, \quad \|R(\lambda)\| = \| R(\lambda,n) \|_\infty.
\end{equation}

To reformulate the assumptions (\ref{Ln}) and (\ref{Rn})  by $L^2(\S^2)$-operator norms conditions, we observe that the selfadjoint operator $|\DS|$ acts as multiplication by $n$ on each generalized spherical harmonics $Y_{kn}$ in the Hilbert decomposition $L^2(\S^2, \C^2) = \oplus_{n,k \in I} \C^2 \otimes Y_{kn}$. We still denote by $|\DS|$ the restriction of this operator to $\mathfrak{l} = L^2(\S^2, \C)$ and thus, $|\DS|$ acts as multiplication by $n$ on each generalized spherical harmonics $Y^1_{kn}$ or $Y^2_{kn}$ in the two Hilbert decompositions $\mathfrak{l} = \oplus_{n,k \in I} \C \otimes Y^j_{kn}, \ j=1,2$.

Now, let $\lambda \in \R$ and $0 < B < \min(A,\tilde{A})$. Then, using (\ref{LinkTLR}), the assumptions (\ref{Ln}) and (\ref{Rn}) for the reflection coefficients can be written as
\begin{equation} \label{L-OpNorm}
  (\ref{Ln}) \quad \Longleftrightarrow \quad \left\| e^{2B|\DS|} \left(L(\lambda) - \tilde{L}(\lambda) \right) \right\|_{\B(\mathfrak{l})} = O(1),
\end{equation}
and
\begin{equation} \label{R-OpNorm}
  (\ref{Rn}) \quad \Longleftrightarrow \quad \left\| e^{2B|\DS|} \left(R(\lambda) - \tilde{R}(\lambda) \right) \right\|_{\B(\mathfrak{l})} = O(1).
\end{equation}
Similarly, let $B > \max(A,\tilde{A})$. Then we get for the assumption (\ref{Tn}) on the transmission coefficients the equivalences
\begin{eqnarray} \label{T-OpNorm}
  (\ref{Tn}) \ & \Longleftrightarrow & \ \left\| e^{2B|\DS|} \left(T_L(\lambda) - \tilde{T_L}(\lambda) \right) \right\|_{\B(\mathfrak{l})} = O(1), \\
             & \Longleftrightarrow & \ \left\| e^{2B|\DS|} \left(T_R(\lambda) - \tilde{T_R}(\lambda) \right) \right\|_{\B(\mathfrak{l})} = O(1). \nonumber
 \end{eqnarray}

\Section{Addendum on the inverse scattering problem from the transmission coefficients $T(\lambda,n)$}

In \cite{DN3}, Theorem 1.1, it is claimed that the knowledge of the transmission coefficients $T(\lambda,n)$ for a fixed $\lambda \ne 0$ and for all $n \in \mathcal{L}$ where $\mathcal{L}$ is a subset of $\N$ satisfying the M\"untz condition
$$
  \sum_{n \in \mathcal{L}} \frac{1}{n} = +\infty,
$$
also determines uniquely the function $a(x)$ up to a translation. The crucial ingredient of the proof can be found in the Proposition 3.13 of \cite{DN3} which states that \\

"If $T(\lambda,n) = \tilde{T}(\lambda,n)$ for all $n \in \mathcal{L}$, then the corresponding reflection coefficients $L(\lambda,n)$ and $\tilde{L}(\lambda,n)$ (resp. $R(\lambda,n)$ and $\tilde{R}(\lambda,n)$) coincide up to a multiplicative constant". \\

The proof of this result given in \cite{DN3} is unfortunately incomplete. In fact, this last point is not so clear and could even be false. We shall try in this Appendix to give some insights of what happens when we try to determine the metric from the transmission coefficient $T(\lambda,n)$.

We first give a correct version of the above result that is weaker than the Proposition 3.13. given in \cite{DN3}.

\begin{prop} \label{Transmission1}
  Let $(\Sigma, g)$ ans $(\tilde{\Sigma}, \tilde{g})$ be two SSAHM whose metrics depend on the functions $a(x)$ and $\tilde{a}(x)$ satisfying the assumptions (\ref{RegulA}) - (\ref{AsympA}). For a fixed energy $\lambda \ne 0$, consider the corresponding countable family of transmission coefficients $T(\lambda,n)$ and $\tilde{T}(\lambda,n)$ for all $n \in \N^*$. Consider also a subset $\mathcal{L}$ of $\N^*$ that satisfies a M\"untz condition
$\displaystyle\sum_{n \in \mathcal{L}} \frac{1}{n} = \infty$. Assume that
$$
  T(\lambda,n) = \tilde{T}(\lambda,n), \quad \forall n \in \mathcal{L}.
$$
Then
$$
  T(\lambda,z) = \tilde{T}(\lambda,z), \quad \forall z \in \C.
$$
Assume moreover that $\frac{1}{\kappa_+} + \frac{1}{\kappa_-} < 0$.
\begin{itemize}
\item If $\frac{1}{\tilde{\kappa_+}} + \frac{1}{\tilde{\kappa_-}} < 0$, there exists a rational function $g(z)$ such that
$$
  L(\lambda,z) = g(z) \tilde{L}(\lambda,z), \quad \forall z \in \C.
$$
\item If $\frac{1}{\tilde{\kappa_+}} + \frac{1}{\tilde{\kappa_-}} > 0$ and $\left( \frac{\tilde{a_-}}{a_+} \right)^{\frac{i \lambda}{\kappa_+}} = \left( \frac{\tilde{a_+}}{a_-} \right)^{\frac{i \lambda}{\kappa_-}}$, there exists a rational function $h(z)$ such that
$$
  L(\lambda,z) = h(z) \tilde{R}(\lambda,z), \quad \forall z \in \C.
$$
\end{itemize}
\end{prop}

\begin{proof}
By definition of the transmission coefficients and using Corollary 3.9. and Theorem 3.10. in \cite{DN3}, our assumption implies that
$$
  T(\lambda,z) = \tilde{T}(\lambda,z), \quad \forall z \in \C,
$$
or equivalently
\begin{equation} \label{rr0}
  a_{L1}(\lambda,z) = \tilde{a}_{L1}(\lambda,z), \quad \forall z \in \C.
\end{equation}

Now, we set ${\displaystyle{f(z)= \frac{a_{L3}(\lambda,z)}{z}}}$. Using that $a_{L3}(\lambda,0)=0$, we see that
$f(z)$ is an even entire function of order $1$ thanks to Lemma \ref{AL-Analytic}. Thus, we can write $f(z)=g(z^2)$ where $g$ is an entire
function of order $\frac{1}{2}$. Using Hadamard's factorization Theorem, we obtain the following expression for $f$
\begin{equation} \label{fz}
f(z) = G \, z^{2m}\ \prod_{n=1}^\infty \Big( 1 - \frac{z^2}{z_n^2} \Big),
\end{equation}
where $2m$ is the multiplicity of $0$, $G$ is a constant and the $z_n$ are the zeros of $f$ belonging to $C^+ = \{z \in \C, \ \Im(z) > 0, \ \textrm{or} \ \Im(z) = 0, \ \Re(z) > 0\}$ counted according to their multiplicity. From (\ref{rr0}) and Lemma \ref{AL-Analytic}, (iii), we have
$$
 f(z)\overline{f(\bar{z})} = \tilde{f}(z)\overline{\tilde{f}(\bar{z})},
$$
where ${\displaystyle{\tilde{f}(z)= \frac{\tilde{a_{L3}}(\lambda,z)}{z}}}$. Thus we get
$$
  \mid G\mid^2 \ z^{4m} \ \prod_{n=1}^\infty \Big( 1 -\frac{z^2}{z_n^2} \Big)
\Big( 1 -\frac{z^2}{\bar{z_n}^2} \Big)  = \ \mid \tilde{G}\mid^2 \ z^{4\tilde{m}} \ \prod_{n=1}^\infty \Big( 1 -\frac{z^2}{\tilde{z_n}^2} \Big)
\Big( 1 -\frac{z^2}{\bar{\tilde{z_n}}^2} \Big).
$$
It follows that $\mid G \mid =\mid \tilde{G}\mid , \ m=\tilde{m}$ and
\begin{equation} \label{rr1}
z_n = \pm \tilde{z_n} \ \textrm{or} \ z_n = \pm \overline{\tilde{z_n}}, \quad \forall n \in \N^*.
\end{equation}

\begin{remark}
  1) The equation (\ref{rr1}) is where we made an error in \cite{DN3}, Proposition 3.13. since we asserted that
	$$
	  z_n = \tilde{z_n}, \quad \forall n \in \N^*.
	$$
	2) If $\Im(z_n) > 0$, then we must have
	\begin{equation} \label{rr2}
z_n = \tilde{z_n} \ \textrm{or} \ z_n = - \overline{\tilde{z_n}}, \quad \forall n \in \N^*.
\end{equation}
Hence, the zeros $z_n$ and $\tilde{z_n}$ with positive imaginary parts coincide \emph{up to "-" complex conjugation}.

On the other hand, if $\Im(z_n) = 0$ and $\Re(z_n) > 0$, then
\begin{equation} \label{rr3}
  z_n = \tilde{z_n}, \textrm{or} \ z_n = \overline{\tilde{z_n}},
\end{equation}
holds.
\end{remark}

In some cases, we can prove that the large zeros $z_n$ and $\tilde{z_n}$ coincide using the asymptotics of $a_{L3}(\lambda,z)$ for large $z$ in the complex plane. We shall use

\begin{lemma} \label{LargeAsymp-aL3}
  For $|z|$ large in the complex plane, we have
	$$\begin{array}{rl}
a_{L1}(\lambda,z) = &  \frac{1}{2\pi} \left( - \frac{\kappa_+}{a_+} \right)^{\frac{i \lambda}{\kappa_+}} \left( \frac{\kappa_-}{a_-} \right)^{-\frac{i\lambda}{\kappa_-}} \Gamma(1-\nu_+) \Gamma(1-\mu_-) \left( \frac{z}{2} \right)^{i \lambda \left(\frac{1}{\kappa_-}-\frac{1}{\kappa_+} \right)} \\
& \times \left( e^{zA} + e^{-zA} e^{-\mathrm{sg}(\mathrm{Im}(z)) \pi \lambda \left(\frac{1}{\kappa_+} - \frac{1}{\kappa_-} \right)} \right) \ \left( 1 + O(\frac{1}{z}) \right),\\

a_{L3}(\lambda,z) = & \frac{i}{2\pi} \left( - \frac{\kappa_+}{a_+} \right)^{\frac{i \lambda}{\kappa_+}} \left( \frac{\kappa_-}{a_-} \right)^{\frac{i \lambda}{\kappa_-}} \Gamma(1-\nu_+) \Gamma(1-\nu_-) \left( \frac{z}{2} \right)^{-i\lambda \left(\frac{1}{\kappa_-}+\frac{1}{\kappa_+} \right)} \\
& \times \left( e^{zA} + e^{-zA} e^{\mathrm{sg}(\mathrm{Im}(z)) i \pi \left(1+ i\lambda \left( \frac{1}{\kappa_+} + \frac{1}{\kappa_-} \right) \right)} \right) \ \left( 1 + O(\frac{1}{z}) \right).
\end{array}
$$
\end{lemma}
\begin{proof}
  We refer to \cite{DaKaNi, FY} where similar asymptotics have been obtained.
\end{proof}

Using Rouche's Theorem and a standard argument (see \cite{FY}), we obtain from Lemma \ref{LargeAsymp-aL3} the following asymptotics for the large zeros $z_n$ with positive imaginary part.

\begin{coro} \label{Asymp-Large-zn}
  There exists $p \in \Z$ such that for large $n$, we have
	$$
	  z_n = i \frac{\pi}{A} ( n + p) - \frac{\lambda \pi}{2A} \left( \frac{1}{\kappa_+} + \frac{1}{\kappa_-} \right) + O(\frac{1}{n}).
	$$
\end{coro}
We conclude from Corollary \ref{Asymp-Large-zn} and the previous Remark that there exists $N \in \N$ large enough such that for all $n > N$, we have
$$
  z_n = \tilde{z_n} \ \textrm{or} \ z_n = - \overline{\tilde{z_n}}.
$$

Assume from now on that $\frac{1}{\kappa_+} + \frac{1}{\kappa_-} < 0$. We conclude that the large zeros of $a_{L3}(\lambda, z)$ with positive imaginary part are located in the quadrant $I = \{z \in \C, \ \Re(z) > 0, \ \Im(z) > 0 \}$. By parity, the zeros with negative imaginary part are located in the quadrant $III = \{z \in \C, \ \Re(z) < 0, \ \Im(z) < 0 \}$.

Since the $\tilde{z_n}$'s with positive imaginary part also satisfy the asymptotics in Corollary \ref{Asymp-Large-zn}, we get the following dichotomy.
\begin{itemize}
\item If $\frac{1}{\tilde{\kappa}_+} + \frac{1}{\tilde{\kappa}_-} < 0$, then the zeros $\tilde{z_n}$'s with positive imaginary part are located in the quadrant $I$. Hence, using (\ref{rr2}), we have the following. There exists a $N \in \N$ such that
\begin{equation} \label{rr4}
  z_n = \tilde{z_n}, \quad \forall n > N.
\end{equation}
Using (\ref{rr0}), Lemma \ref{LargeAsymp-aL3} and Corollary \ref{Asymp-Large-zn}, we get in this case
$$
  \frac{1}{\tilde{\kappa}_-} - \frac{1}{\tilde{\kappa}_+} = \frac{1}{\kappa_-} + \frac{1}{\kappa_+}, \quad \frac{1}{\tilde{\kappa}_-} + \frac{1}{\tilde{\kappa}_+} = \frac{1}{\kappa_-} + \frac{1}{\kappa_+},
$$
which gives
$$
  \tilde{\kappa}_- = \kappa_-, \quad \tilde{\kappa}_+ = \kappa_+.
$$
Also, we use (\ref{fz}) and (\ref{rr4}) to obtain
\begin{equation} \label{rr5}
  f(z) = \frac{G}{\tilde{G}} \prod_{n = 1}^N \frac{\Big( 1 -\frac{z^2}{z_n^2} \Big)}{\Big( 1 -\frac{z^2}{\tilde{z_n}^2} \Big)} \tilde{f}(z).
\end{equation}
Finally, denote by $E_N = \left\{n \in \{1,\dots,N\}, \ z_n \ne \tilde{z_n}, \ \textrm{and} \ z_n \ne -\tilde{z_n} \right\}$. Then we obtain from (\ref{rr5})
\begin{equation} \label{rr6}
  f(z) = \frac{G}{\tilde{G}} \prod_{n \in E_N} \frac{\Big( 1 -\frac{z^2}{\overline{\tilde{z_n}}^2} \Big)}{\Big( 1 -\frac{z^2}{\tilde{z_n}^2} \Big)} \tilde{f}(z).
\end{equation}
Denoting by $g(z)$ the rational function $g(z) = \frac{G}{\tilde{G}} \prod_{n \in E_N} \frac{\Big( 1 -\frac{z^2}{\overline{\tilde{z_n}}^2} \Big)}{\Big( 1 -\frac{z^2}{\tilde{z_n}^2} \Big)}$, we finally get
$$
  a_{L3}(\lambda,z) = g(z) \tilde{a_{L3}}(\lambda,z),
$$
and thus
$$
  L(\lambda,z) = g(z) \tilde{L}(\lambda,z).
$$

\item If $\frac{1}{\tilde{\kappa}_+} + \frac{1}{\tilde{\kappa}_-} > 0$, then the zeros $\tilde{z_n}$'s with positive imaginary part are located in the quadrant $II = \{ z \in \C, \ \Re(z) < 0, \ \Im(z) > 0 \}$. Hence, using (\ref{rr2}), we have the following. There exists a $N \in \N$ such that
\begin{equation} \label{rr7}
  z_n = - \overline{\tilde{z_n}}, \quad \forall n > N.
\end{equation}
Using (\ref{rr0}), Lemma \ref{LargeAsymp-aL3} and Corollary \ref{Asymp-Large-zn}, we get in this case
$$
  \frac{1}{\tilde{\kappa}_-} - \frac{1}{\tilde{\kappa}_+} = \frac{1}{\kappa_-} + \frac{1}{\kappa_+}, \quad \frac{1}{\tilde{\kappa}_-} + \frac{1}{\tilde{\kappa}_+} = - \frac{1}{\kappa_-} - \frac{1}{\kappa_+},
$$
which gives
$$
  \tilde{\kappa}_- = -\kappa_+, \quad \tilde{\kappa}_+ = -\kappa_-.
$$
Using again (\ref{rr0}) and the asymptotics of $a_{L1}(\lambda,z)$ from Lemma \ref{LargeAsymp-aL3}, we get the necessary condition
$$
  \left( \frac{\tilde{a_-}}{a_+} \right)^{\frac{i \lambda}{\kappa_+}} = \left( \frac{\tilde{a_+}}{a_-} \right)^{\frac{i \lambda}{\kappa_-}}.
$$
Also, we use (\ref{fz}) and (\ref{rr7}) to obtain
\begin{equation} \label{rr8}
  f(z) = \frac{G}{\overline{\tilde{G}}} \prod_{n = 1}^N \frac{\Big( 1 -\frac{z^2}{z_n^2} \Big)}{\Big( 1 -\frac{z^2}{\overline{\tilde{z_n}}^2} \Big)} \overline{\tilde{f}(\bar{z})}.
\end{equation}
Finally, denote by $F_N = \{n \in \{1,\dots,N\}, \ z_n \ne \overline{\tilde{z_n}}, \ \textrm{and} \ z_n \ne -\overline{\tilde{z_n}} \}$. Then we obtain from (\ref{rr5})
\begin{equation} \label{rr9}
  f(z) = \frac{G}{\overline{\tilde{G}}} \prod_{n \in F_N} \frac{\Big( 1 -\frac{z^2}{\tilde{z_n}^2} \Big)}{\Big( 1 -\frac{z^2}{\overline{\tilde{z_n}}^2} \Big)} \overline{\tilde{f}(\bar{z})}.
\end{equation}
Denoting by $h(z)$ the rational function $h(z) = \frac{G}{\overline{\tilde{G}}} \displaystyle\prod_{n \in F_N} \frac{\Big( 1 -\frac{z^2}{\tilde{z_n}^2} \Big)}{\Big( 1 -\frac{z^2}{\overline{\tilde{z_n}}^2} \Big)}$, we finally get
$$
  a_{L3}(\lambda,z) = h(z) \overline{\tilde{a_{L3}}(\lambda,\bar{z})} = h(z) \tilde{a_{L2}}(\lambda,z),
$$
and thus
$$
  L(\lambda,z) = - h(z) \tilde{R}(\lambda,z).
$$
\end{itemize}

Both above cases prove the results stated in the Proposition.
\end{proof}

Even in the case when the reflection coefficients $L(\lambda,z)$ and $\tilde{L}(\lambda,z)$ (resp. $R(\lambda,z)$ and $\tilde{R}(\lambda,z)$) coincide up to a \emph{rational function} in the $z$ variable, we cannot conclude from this fact the result stated in \cite{DN3}, that is the uniqueness of the function $a(x)$ and $\tilde{a}(x)$ up to a translation. This question remains thus open and we conjecture that this is false. We refer to the last Section of \cite{DaKaNi} for more details about this point in a similar and more general model.

What we can prove however is the following weaker statement.
\begin{prop}\label{unicitetransmission}
Assume that
$$
	\frac{1}{\kappa_+} + \frac{1}{\kappa_-} < 0, \quad \frac{1}{\tilde{\kappa}_+} + \frac{1}{\tilde{\kappa}_-} < 0.
$$
	Let $\mathcal{L}$ be a subset of $\N$ such that $\displaystyle\sum_{n \in \mathcal{L}} \frac{1}{n} = \infty$. Assume that
$$
  T(\lambda,n) = \tilde{T}(\lambda,n), \quad \forall n \in \mathcal{L}.
$$
Assume also that
\begin{equation}\label{rr10}
  L(\lambda,k) = \tilde{L}(\lambda,k),
\end{equation}	
for a finite but large enough number of indices $k \in \N$. Then there exists a constant $\sigma \in \R$ such that
$$
  \tilde{a}(x) = a(x + \sigma).
$$
In consequence, the two SSAHM $(\Sigma,g)$ ans $(\tilde{\Sigma},\tilde{g})$ coincide up to isometries.
\end{prop}

\begin{proof}
From Proposition \ref{Transmission1}, we know that there exists a rational function $g(z)$ such that
$$
  L(\lambda,z) = g(z) \tilde{L}(\lambda,z).
$$
From (\ref{rr10}), we infer that $g(z) = 1$ for all $z\in \C$ and thus
$$
 L(\lambda,z) = \tilde{L}(\lambda,z).
$$
Hence the Proposition is proved using Theorem 1.1. in \cite{DN3}.
\end{proof}


\end{document}